%% file: EC_camera.tex
\documentclass[runningheads]{llncs}

\usepackage{epsfig}
\usepackage{tikz}
\usepackage{inputenc}
\usepackage{amssymb}
\usepackage{amsmath}
\usepackage{array}
\usepackage{listing}
\usepackage{underscore}
\usepackage{graphicx}
\usepackage{eso-pic}
\usepackage{fancyhdr}
\usepackage{algorithmic}
\usepackage{algorithm}
\usepackage{hyperref}  
\usepackage[all]{hypcap}
\usepackage{float}
\usepackage{pdfpages}
 \usepackage{hyperref}
\usepackage{hyperref}
\usepackage[
	lambda,
	operators,
        general,
	advantage, 
	sets,     
	adversary,
	landau,
	probability,
	notions,	
	logic,
	ff,
	mm, 
	primitives,
	events,
	complexity,
	asymptotics,
	keys]{cryptocode}

\usepackage{csquotes}

\usepackage{dashbox}
\usepackage{todonotes}

\usepackage{url}

\usetikzlibrary{shapes.callouts}

\usepackage{listings}

\usepackage{trace}

\usepackage{makeidx}
\usepackage{mdframed}
\usepackage[noadjust]{cite}
\makeindex

\definecolor{mygreen}{rgb}{0,0.6,0}
\definecolor{mygray}{rgb}{0.1,0.1,0.1}
\definecolor{mymauve}{rgb}{0.58,0,0.82}

\input{macro}
\input{preamble1}

\usepackage{authblk}

\title{Compactness of Hashing Modes and Efficiency beyond Merkle Tree}
\author{Elena Andreeva \inst{1}
  \and Rishiraj Bhattacharyya\inst{2}
  \and Arnab Roy\inst{3}}
\institute{Technical University of Vienna, Austria\and
  NISER, HBNI, India\and
  University of Klagenfurt, Austria\\
\texttt{elena.andreeva@tuwien.ac.at, rishirajbhattacharyya@protonmail.com, arnab.roy@aau.at}}

\begin{document}
\maketitle
\begin{abstract}

\input{abstract.tex}
  \end{abstract}

\input{intro.tex}

\section{Notation and Preliminaries}
\label{sec:notat-prel}

\input{prelim.tex}

\section{Compactness: Normalizing Efficiency for Optimally Secure Constructions}
\label{sec:defcompact}
\input{compact.tex}

\input{collisionresistancealt.tex}

\section{Achieving Indifferentiability Efficiently}
\label{sec:effic-vs-indiff}
Below we first consider the basic $\cmt$ compression function and analyze its security with respect to the  indifferentiability notion. We show that while $\cmt$ fails to achieve indifferentiability, a simple modification can restore the indifferentiability. We call that modified tree $\cmt^{+}$ mode construction.  $\cmt^{+}$ mode is the merge of two $\cmt$ modes (trees), not necessarily of the same height $\ell\geq 2$ each, and feeding their inputs to a final compression function (omitting the final message injection and feedforward). 

\subsection{Indifferentiability attack against $\cmt$ mode}
\label{sec:mmt-does-not}
\input{indiffattack.tex}

\subsection{Almost Fully Compact and Indifferentiable $\cmt^{+}$ Mode}
\label{sec:full-indiff-price}
\input{indiffmt1.tex}
\section{Efficiency and Applications}
\label{sec:comp-and-appl}
 \input{applicationv1.tex}


\section{Discussion and Conclusions}
\label{sec:discuss}

\input{discussion.tex}

\noindent
\textbf{Acknowledgements} We thank Martijn Stam for reading an earlier version of the draft and providing valuable comments. We would also like to thank Markulf Kohlweiss for discussion (during Arnab's visit to the University of Edinburgh in 2019) on the zksnarks and other applications of this work. We sincerely thank the reviewers of Eurocrypt 2021, Asiacrypt 2020 for their insightful comments. We are grateful to the reviewers of Crypto 2020 for their suggestions to extend our previous work for generalized tree-like hash.

\bibliographystyle{plain}
\bibliography{abbrev3,crypto,other}



\end{document}

%% file: macro.tex
\newtheorem{prop}[theorem]{Proposition}
\newtheorem{construction}[theorem]{Construction}
 
\newcommand{\twopartdef}[4]
{
	\left\{
		\begin{array}{ll}
			#1 & \mbox{if } #2 \\
			#3 & \mbox{if } #4
		\end{array}
	\right.
}

\newcommand{\algo}[1]{\mbox{\tt #1}}

\newcommand{\getsr}{\gets_{\mbox{\tiny R}}}

\newcommand{\Adv}{\mbox{\bf Adv}}

%% file: preamble1.tex
\usepackage{graphicx}
\usepackage{tikz}
\usepackage{tikzscale}
\usepackage{cleveref}
\usepackage{cite}

\usepackage{subcaption}
\usepackage{algorithm,algorithmic,verbatim}
\usetikzlibrary{shapes,snakes,positioning,calc}
\usetikzlibrary{arrows,shadows}

\usetikzlibrary{calc}
\usetikzlibrary{positioning}
\usetikzlibrary{decorations.pathreplacing}
\usetikzlibrary{shapes.geometric}

\tikzset{XOR/.style={draw, color = red, fill = white, circle, append after command={
        [shorten >=\pgflinewidth, shorten <=\pgflinewidth, thick]
        ($(\tikzlastnode.north)+(0,-0.5ex)$) edge[thick] ($(\tikzlastnode.south)+(0,0.5ex)$)
        ($(\tikzlastnode.east)+(-0.5ex,0)$) edge[thick] ($(\tikzlastnode.west)+(0.5ex,0)$)
        }
    }}

\tikzset{compres/.style={
  draw, trapezium, trapezium angle=285, draw, inner xsep=13pt,inner ysep=6pt, outer sep=0pt,
  minimum height=1.81mm
}}
\tikzset{large abr/.style={draw,isosceles triangle,minimum width=27mm, minimum height=36mm,shape border rotate=-90}}
\tikzset{very large abr/.style={draw,isosceles triangle,minimum width=33mm, minimum height=43mm,shape border rotate=-90}}

\newcommand{\setfont}[1]{{\mathcal{#1}}}
\newcommand{\algfont}[1]{{\mathsf{#1}}}

\renewcommand{\hash}{\ensuremath{H}}

\newcommand{\cmt}{\textsf{ABR}}

\newcommand{\msgspace}{\setfont{M}}

\newcommand{\hashspace}{\setfont{Y}}

\newcommand{\algA}{\algfont{A}}

\renewcommand{\th}{^{\mbox{th}}}
\newcommand{\notionfont}[1]{\mathrm{#1}}

\newcommand{\Coll}{\ensuremath{\notionfont{Coll}}}
\newcommand{\Indiff}{\ensuremath{\notionfont{Indiff}}}

\renewcommand{\Adv}[2]{\mathbf{Adv}^{#1}_{#2}}

\def\emptystring{\varepsilon}

\renewcommand{\Dom}{\mathsf{Dom}}
\renewcommand{\Rng}{\mathsf{Rng}}

\newcommand{\N}{{{\mathbb N}}}

\def\getsr{\stackrel{{\scriptscriptstyle\$}}{\leftarrow}}

\newcommand{\Prob}[1]{{\Pr\left[\,{#1}\,\right]}}

\newcommand{\CondProb}[2]{{\Pr}\left[\: #1\::\:#2\:\right]}

\newcommand{\Expect}[1]{{{\mathbf{E}}\left[\,{#1}\,\right]}}

%% file: abstract.tex
We revisit the classical problem of designing optimally efficient cryptographically secure hash functions. Hash functions are traditionally designed via applying modes of operation on primitives with smaller domains. The results of Shrimpton and Stam (ICALP 2008), Rogaway and Steinberger (CRYPTO 2008), and Mennink and Preneel (CRYPTO 2012) show how to achieve optimally efficient designs of $2n$-to-$n$-bit compression functions from non-compressing primitives with asymptotically optimal $2^{n/2-\epsilon}$-query collision resistance. Designing optimally efficient and secure hash functions for larger domains ($> 2n$ bits) is still an open problem.  

To enable efficiency analysis and comparison across hash functions built from primitives of different domain sizes, in this work we propose the new \textit{compactness} efficiency notion. It allows us to focus on asymptotically optimally collision resistant hash function and normalize their parameters based on Stam's bound from CRYPTO 2008 to obtain maximal efficiency.

We then present two tree-based modes of operation as a design principle for compact, large domain, fixed-input-length hash functions.
\begin{enumerate}
  \item Our first construction is an \underline{A}ugmented \underline{B}inary T\underline{r}ee (\cmt) mode. The design is a $(2^{\ell}+2^{\ell-1} -1)n$-to-$n$-bit hash function making a total of $(2^{\ell}-1)$ calls to $2n$-to-$n$-bit compression functions for any $\ell\geq 2$. Our construction is optimally compact with asymptotically (optimal) $2^{n/2-\epsilon}$-query collision resistance in the ideal model. For a tree of height $\ell$, in comparison with Merkle tree, the $\cmt$ mode processes additional $(2^{\ell-1}-1)$ data blocks making the same number of internal compression function calls.
  \item   With our second design we focus our attention on the indifferentiability security notion. While the $\cmt$ mode achieves collision resistance, it fails to achieve indifferentiability from a random oracle within $2^{n/3}$ queries. $\cmt^{+}$ compresses only $1$ less data block than $\cmt$ with the same number of compression calls and achieves in addition indifferentiability  up to $2^{n/2-\epsilon}$ queries.
  \end{enumerate}
Both of our designs are closely related to the ubiquitous Merkle Trees and have the potential for real-world applicability where the speed of hashing is of primary interest.

%% file: intro.tex
\section{Introduction}
Hash functions are fundamental cryptographic building blocks. The art of designing a secure and efficient hash function is a classical problem in
cryptography. Traditionally, one designs a hash function in two steps. In the first, one constructs a \textit{compression function} that maps fixed
length inputs to fixed and usually smaller length outputs.  In the second step, a \textit{domain extending} algorithm is designed that allows longer messages to be mapped to a fixed-length output via a sequence of calls to the underlying compression functions.

Most commonly compression functions are designed based on block ciphers and
permutations~\cite{C:BlaRogShr02,EC:BlaCocShr05,AC:RisShr07,C:RogSte08,EPRINT:BDPA11,SAC:BDPV11}. For a long time block ciphers were the most popular
primitives to build a compression function and the classical constructions of MD5 and SHA1, SHA2 hash functions are prominent examples of that approach. In the light
of the SHA3 competition, the focus has shifted to permutation~\cite{EC:BDPV08} or fixed-key blockcipher-based \cite{DBLP:conf/isw/AndreevaMP10,SCN:AndMenPre10}
compression functions.  Classical examples of domain extending algorithms are the Merkle--Damg{\aa}rd~\cite{C:Merkle89a,C:Damgaard89b} (MD) domain extender and the Merkle tree~\cite{DBLP:conf/sp/Merkle80} which underpins numerous cryptographic applications. Most recently, the Sponge construction~\cite{EC:BDPA13} that is used in SHA-3 has come forward as a domain extender~\cite{SAC:BDPV11, CHES:BKLTVV11, DBLP:journals/ijisec/AndreevaMP12,DBLP:journals/iacr/RivestS16} method for designs which directly call a permutation.

\noindent\underline{\textsc{Efficiency of Hash Design: Lower Bounds.}} Like in all cryptographic primitives, the design of a hash function is a trade-off between efficiency and security. Black, Cochran, and Shrimpton~\cite{EC:BlaCocShr05} were the first to formally analyze the security-efficiency trade-off of compression functions, showing that a $2n$-to-$n$-bit compression function making a single call to a fixed-key $n$-bit block cipher can not achieve collision resistance. Rogaway and Steinberger~\cite{EC:RogSte08} generalized the result to show that any $mn$-to-$ln$ bit compression function making $r$ calls to $n$-bit permutations is susceptible to a collision attack in $(2^n)^{1-\frac{m-l/2}{r}}$ queries, provided the constructed compression function satisfies a ``collision-uniformity'' condition. Stam~\cite{C:Stam08} refined this result to general hash function constructions and conjectured: if any $m+s$-to-$s$-bit hash function is designed using $r$ many $n+c$-to-$n$-bit compression functions, a collision on the hash function can be found in $2^{\frac{nr+cr-m}{r+1}}$ queries. This bound is known as the Stam's bound and it was later proven in two works by Steinberger~\cite{DBLP:conf/eurocrypt/Steinberger10} and by Steinberger, Sun and Yang~\cite{C:SteSunYan12}.    

\noindent\underline{\textsc{Efficiency of Hash Design: Upper Bounds.}} The upper bound results matching Stam's bound focused on $2n$-to-$n$-bit constructions from $n$-bit non-compressing primitives. In~\cite{ICALP:ShrSta08}, Shrimpton and Stam showed a (Shrimpton-Stam) construction based on three $n$-to-$n$-bit functions achieving asymptotically birthday bound collision resistance in the random oracle model. Rogaway and Steinberger~\cite{C:RogSte08} showed hash constructions using three $n$-bit permutations matching the bound of~\cite{EC:RogSte08} and assuming the ``uniformity condition'' on the resulting hash construction.   In~\cite{C:MenPre12}, Mennink and Preneel generalized these results and identified four equivalence classes of $2n$-to-$n$-bit compression functions from $n$-bit permutations and XOR operations, achieving collision security of the birthday bound asymptotically in the random permutation model. 

In comparison, upper bound results for larger domain compressing functions have been
scarce. The only positive result we are aware of is by Mennink and Preneel~\cite{BartnBart}. In~\cite{BartnBart}, the authors considered generalizing the Shrimpton-Stam construction to get $m+n$-to-$n$-bit hash function from $n$-bit primitives for $m>n$, and showed  $n/3$-bit collision security in the random oracle model. For all practical purposes the following question remains open.\\
\noindent{\textit{If an $m+n$-to-$n$-bit hash function is designed using $r$ many $n+c$-to-$n$-bit compression functions, is there a construction with collision security matching Stam's bound when $m>n$?}}  \\
\noindent\underline{\textsc{Beyond Collision Resistance: Indifferentiability.}} {\em Collision resistance} is undoubtedly the most commonly mandated security property for a cryptographic hash function. Naturally, all the hash function design principles and respective efficiencies are primarily targeting to achieve collision resistance. More recently, for applications of hash functions as replacement of random oracles in higher-level cryptographic schemes or protocols, the notion of indifferentiability has also gained considerable traction. The strong notion of indifferentiability from a random oracle (RO)  by Maurer, Renner and Holenstein~\cite{TCC:MauRenHol04} has been adopted to prove the security of hash functions when the internal primitives (compression functions, permutations etc.) are assumed to be ideal (random oracle, random permutation, etc.). An important advantage of the indifferentiability from a random oracle notion is that it  implies multiple security notions (in fact, all the notions satisfied by a random oracle in a single stage game) simultaneously up to the proven indifferentiability bound. The question of designing  an optimally efficient hash function naturally gets extended also to the indifferentiability setting.\\ \textit{\textit{If an $m+n$-to-$n$-bit hash function is designed using $r$ many $n+c$-to-$n$-bit compression functions, is there a construction with indifferentiability security matching Stam's bound when $m>n$?}} Note that, a collision secure hash function matching Stam's bound may not imply the  indifferentiability notion up to the same bound.

\subsection{Our Results}
\label{sec:our-result}
\underline{\textsc{New measure of efficiency.}}
Comparing efficiency of hash functions built from primitives of different domain sizes is a tricky task. In addition to the message size and the number of calls to underlying primitives, one needs to take into account the domain and co-domain/range sizes of the underlying primitives. It is not obvious how to scale the notion of rate up to capture these additional parameters. \\
We approach the efficiency measure question  from Stam's bound perspective. We say an $m+s$-to-$s$-bit hash function construction designed using $r$ many $n+c$-to-$n$-bit compression functions is optimally efficient if Stam's bound is tight, that is one can prove that asymptotically at least $2^{\frac{nr+cr-m}{r+1}}$ queries are required to find a collision. Notice that the value in itself can be low (say $2^{{s/4}}$), but given the proof, we can argue that the parameters are \emph{optimal} for that security level.\\
Given that the collision-resistance requirement for a hash function is given by the birthday bound ($2^{s/2}$ queries), we can say that a hash function construction achieves optimal security-efficiency trade-off if $\frac{nr+cr-m}{r+1}=\frac{s}{2}$ and Stam's bound is asymptotically tight. Then one can focus on schemes which achieve the asymptotically optimal collision security, and normalize the efficiency of the construction. We hence propose the notion of \textit{compactness} as the ratio of the parameter $m$ and its optimal value ($\frac{2nr+2cr-sr-s}{2}$) as an efficiency measure of a hash function construction $C$. In Section~\ref{sec:defcompact} we formally define the notion and derive compactness of some popular modes.      

\begin{figure}[t!]
 	\centering
 	\begin{subfigure}[t]{0.4\textwidth}
 		\centering
 		\scalebox{0.7}{\input{figs/mtree.tikz}}
 		\caption{Merkle tree hashing $4$ input messages}
 		\label{fig:merkletree}
 	\end{subfigure}
 	\hspace{1.5cm}
 	\begin{subfigure}[t]{0.4\textwidth}
 		\centering
 		\scalebox{0.7}{\input{figs/abrtree.tikz}}
 		\caption{$\cmt$ mode hashing $5$ input messages}
 		\label{fig:abr-compress}
 	\end{subfigure}
 	\caption{Merkle Tree and $\cmt$ mode for height $\ell=2$}
      \end{figure}
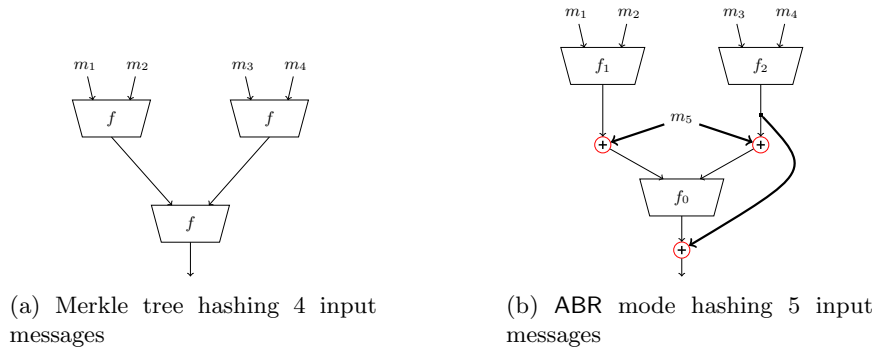

\underline{\textsc{Optimally Compact $\cmt$ Mode.}}
We present a new tree-based mode $\cmt$. $\cmt$ of height $\ell$ implements a $(2^{\ell}+2^{\ell-1}-1)n$-to-$n$-bit function making only $(2^{\ell}-1)$ calls to the underlying $2n$-to-$n$-bit compressing primitives. Assuming the underlying primitives to be independent random oracles, we show that the $\cmt$ mode is collision resistant up to the birthday bound asymptotically. The parameters of $\cmt$ mode achieve maximum compactness. In Section \ref{sec:coll-resist-gener} we formally present the $\cmt$ mode and prove its collision resistance. 

A natural comparison with Merkle tree is in order. We show that Merkle Tree can achieve only $2/3$ of the optimal compactness and thus our mode is significantly more efficient. For a tree of height $\ell$, in comparison to the Merkle tree, the $\cmt$ mode can process an additional $(2^{\ell-1}-1)$ message blocks with the same number of calls to the underlying compression functions.


\underline{\textsc{$\cmt$ does not satisfy Indifferentiability.}}
Our next target is to consider the notion of indifferentiability. Specifically, how does the $\cmt$ compression score in the indifferentiability setting? The primary objective of this question is twofold. If we can prove the $\cmt$ construction with height $\ell=2$ to be indifferentiable from a random oracle up to the birthday bound, then we could use the indifferentiability composition theorem and replace the leaf level compression function of $\cmt$ by $5n$-to-$n$-bit ideal compression function. Then by recursively applying the proof of collision resistency of $\cmt$ with height $\ell=2$, we could extend the collision resistance proof to arbitrary large levels. Secondly, the proof of indifferentiability implies simultaneously all the security notions satisfied by a random oracle in single stage games. Unfortunately, we show that the $\cmt$ mode with height $\ell=2$ does not preserve indifferentiability. We show an indifferentiability attack of order $2^{\frac{n}{3}}$ in~\cref{sec:effic-vs-indiff}. The attack can easily be generalized to $\cmt$ of arbitrary levels.

\noindent\underline{\textsc{Salvaging Indifferentiability.}}
Next, in Section~\ref{sec:full-indiff-price} we propose an almost optimally compact $\cmt^{+}$ mode design which salvages the indifferentiability security (up to birthday bound) of the original $\cmt$ mode.  In principle, our second construction $\cmt^{+}$ (see Fig.~\ref{figs/abrindif.tikz}) tree merges two left and right $\cmt$ mode (of possibly different heights) calls by an independent post-precessor. Using the H-coefficient technique, we prove the indifferentiability of the $\cmt^{+}$  construction up to the birthday bound.

Compared to $\cmt$ mode,  $\cmt^{+}$ compresses 1 less message block for the same number of calls. For large size messages, this gap is extremely small. In comparison to the Merkle Tree, the $\cmt^{+}$ mode, improves the efficiency significantly and still maintains the indifferentiability property.

\subsection{Impact of Our Result} Merkle trees were first published in 1980 by Ralph Merkle~\cite{DBLP:conf/sp/Merkle80} as a way to authenticate large public files. Nowadays, Merkle trees find ubiquitous applications in cryptography, from parallel hashing, integrity checks of large files, long-term storage, signature schemes~\cite{tlssign,INDOCRYPT:BGDDK06,ACNS:BDKOV07,CCS:BHKNRS19}, time-stamping~\cite{Haber91howto}, zero-knowledge proof based protocols~\cite{EC:GGPR13, USENIX:BCTV14}, to anonymous cryptocurrencies~\cite{zcash}, among many others. 
Despite their indisputable practical relevance, for 40 years we have seen little research go into the rigorous investigation of how to optimize their efficiency, and hence we still rely on design principles that may in fact have some room for efficiency optimizations. 

In view of the  wide spread use of Merkle trees, we consider one of the main advantage of our construction as being in:  {\em increased number of message inputs (compared to the classical Merkle tree) while maintaining the same tree height and computational cost (for both root computation and node authentication)}. Our trees then offer more efficient alternatives to Merkle trees in scenarios where the performance criteria is \textit{the number of messages hashed} for: 1. a fixed computational cost -- compression function calls to compute the root, or/and 2. fixed authentication cost -- compression function calls to authenticate a node. 

Regular hashing is naturally one of the first candidates for such an applications. Other potential use cases are hashing on parallel processors or multicore machines, such as authenticating software updates, image files or videos; integrity checks of large files systems, long term archiving~\cite{archv}, memory authentication, content distribution, torrent systems \cite{bittorrent}, etc. A recent application that can benefit from our $\cmt$ or $\cmt^{+}$ mode designs are (anonymous) cryptocurrency applications. We elaborate more on these in Section~\ref{sec:comp-and-appl}.

%% file: figs/mtree.tikz
\begin{tikzpicture}

  \node[compres] (h1) at (1,1) {$f$};

  \node (m1) at (0.5,2) {$m_1$};
  \node (m2) at (1.5,2) {$m_2$};
  \draw[->] (m1) -- (h1.north west);
  \draw[->] (m2) -- (h1.north east);

  \node[compres] (h2) at (4,1) {$f$};

  \node (m3) at (3.5,2) {$m_3$};
  \node (m4) at (4.5,2) {$m_4$};
  \draw[->] (m3) -- (h2.north west);
  \draw[->] (m4) -- (h2.north east);

  \node[compres] (h3) at (2.5,-1) {$f$};
  \draw[->] (h1.south) -- (h3.north west);
  \draw[->] (h2.south) -- (h3.north east);

  \draw[->] (h3.south) -- (2.5,-2);

\end{tikzpicture}

%% file: figs/abrtree.tikz
\begin{tikzpicture}

  \node[compres] (h1) at (1,1) {$f_1$};

  \node (m1) at (0.5,2) {$m_1$};
  \node (m2) at (1.5,2) {$m_2$};
  \draw[->] (m1) -- (h1.north west);
  \draw[->] (m2) -- (h1.north east);

  \node[compres] (h2) at (4,1) {$f_2$};

  \node (m3) at (3.5,2) {$m_3$};
  \node (m4) at (4.5,2) {$m_4$};
  \draw[->] (m3) -- (h2.north west);
  \draw[->] (m4) -- (h2.north east);

  \node[compres] (h3) at (2.5,-1.5) {$f_0$};
  \node[XOR] at (1,-0.5) (x1) {};
  \node[XOR] at (4,-0.5) (x2) {};

  \draw[->] (h1.south) -- (x1);
  \draw[->] (h2.south) -- (x2);

  \draw[->] (x1) -- (h3.north west);
  \draw[->] (x2) -- (h3.north east);

  \node (m5) at (2.5,0) {$m_5$};
  \draw[->, very thick] (m5) -- (x1);
  \draw[->, very thick] (m5) -- (x2);

  \node[XOR] at (2.5,-2.5) (x3) {};

\node[fill=black,inner sep=1pt] (x) at ($ (h2.south)!.5!(x2) $) {};


\draw[->, very thick] (x) .. controls (5,-1) .. (x3);

  \draw[->] (h3.south) -- (x3);
  \draw[->] (x3) -- (2.5, -3);

\end{tikzpicture}

%% file: prelim.tex
 Let $\N = \{0,1,\ldots\}$ be the set of natural numbers and $\{0,1\}^*$ be the set of all bit strings. If $k
\in \N$, then $\{0,1\}^k$ denotes the set of all $k$-bit strings.
The empty string is denoted by $\emptystring$. $[n]$ denotes the set $\{0,1,\cdots,n-1\}$. $f:[r]\times\Dom\to\Rng$ denotes a family of $r$ many functions from $\Dom$ to $\Rng$. {\bf We often use the shorthand $f$ to denote the family $\{f_0,\cdots,f_{r-1}\}$ when the function family is given as oracles}. 

\noindent If~$S$ is a set, then $x \getsr S$ denotes the uniformly random selection of an element from~$S$. We let $y \gets \algA(x)$ and $y \getsr \algA(x)$ be the assignment to~$y$ of the output of a deterministic and randomized algorithm $\algA$, respectively, when run on input $x$.

An \textit{adversary} $\algA$ is an algorithm possibly with access to oracles $\mathcal{O}_1, \ldots, \mathcal{O}_{\ell}$ denoted by $\algA^{\mathcal{O}_1, \ldots , \mathcal{O}_{\ell}}$. The adversaries considered in this paper are computationally unbounded. The complexities of these algorithms are measured solely on the number of queries they make. Adversarial queries and the corresponding responses are stored in a transcript $\tau$.\\  
\textbf{ Hash Functions and Domain Extensions.}
In this paper, we consider Fixed-Input-Length (FIL) hash functions. We denote these by the hash function  $\hash: \msgspace \to \hashspace$ where $\hashspace$ and $\msgspace$ are finite sets of bit strings. For a FIL $\hash$ the domain $\msgspace= \{0,1\}^{N}$ is a finite set of  $N$-bit strings.

 Note that, modelling the real-world functions such as SHA-2 and SHA-3, we consider the hash function to be unkeyed. Typically, a hash function is designed in two steps. First a compression function $f:\msgspace_{f}\to\hashspace$ with small domain is designed. Then one uses a domain extension algorithm $C$, which has a blackbox access to $f$ and implements the hash function $\hash$ for larger domain.
\begin{definition}
A domain extender $C$ with oracle access to a family of compression functions $f:[r]\times\msgspace_{f}\to\hashspace$ is an algorithm which implements the  function $\hash=C^f:\msgspace\to\hashspace$.
\end{definition}

\noindent\textbf{Collision Resistance.}
Our definitions of collision (\Coll) security is given for any general FIL hash function $\hash$ built upon the compression functions $f_i$ for $i \in [r]$ where $f_i$s are modeled as ideal random functions.
Let $\text{Func}(2n, n)$ denote the set of all functions mapping $2n$ bits to $n$ bits. Then, for a fixed adversary $\algA$ and for all $i\in[r]$ where $f_i\getsr \text{Func}(2n, n)$, 
we consider the following definition of collision resistance.
\begin{definition}
  \label{def:CR}
Let $\algA$ be an adversary against $\hash=C^f$. $\hash$ is said to be $(q,\varepsilon)$ collision resistant if for all algorithm $\algA$ making $q$ queries it holds that
\begin{equation*}
\Adv{\Coll}{\hash}(\algA) = \CondProb{M', M \getsr \algA^{f}(\emptystring)}{M \neq M' \textrm{ and } \hash(M)=\hash(M')}\leq \varepsilon.
\end{equation*}
\end{definition}

\noindent\textbf{Indifferentiability.}

In the game of indifferentiability, the distinguisher is aiming to distinguish between two worlds, the \emph{real} world and the {\em ideal} world. In the real world, the distinguisher has oracle access to $(C^{\mathcal{F}},\mathcal{F})$ where  $C^{\mathcal{F}}$ is a construction based on an ideal primitive $\mathcal{F}$. In the ideal world the distinguisher has oracle access to $(\mathcal{G},S^\mathcal{G})$ where $\mathcal{G}$ is an ideal functionality and $S$ is a simulator.

\begin{definition}[\bf Indifferentiability \cite{TCC:MauRenHol04}]
\label{def:indiff} 
A Turing machine $C$ with oracle access to an ideal primitive $\mathcal{F}$ is said to be $(t_A,t_S,q_S,q,\varepsilon)$ indifferentiable (Fig.~\ref{fig:indiff}) from an ideal primitive $\mathcal{G}$ if there exists a simulator $S$ with an oracle access to $\mathcal{G}$ having running time at most $t_S$, making at most $q_S$ many calls to  $\mathcal{G}$ \emph{per invocation}, such that for any adversary $\algA$, with running time $t_A$ making at most $q$ queries, it holds that
\begin{align*}
   \Adv{\Indiff}{(C^{\mathcal{F}},\mathcal{F}),(\mathcal{G},S^\mathcal{G})}{(\algA)} \eqdef \left| \Pr[\algA^{(C^{\mathcal{F}},\mathcal{F})} = 1] - \Pr[\algA^{(\mathcal{G},S^\mathcal{G})} = 1] \right| \leq \varepsilon
\end{align*}
$C^{\mathcal{F}}$ is computationally indifferentiable from
$\mathcal{G}$ if $t_A$ is bounded above by some polynomial in the security parameter $k$ and $\varepsilon$ is a negligible function of $k$.
\end{definition}
In this paper, we consider an information-theoretic adversary implying $t_A$ is unbounded. We derive the advantage in terms of the query complexity of the distinguisher. The composition theorem of indifferentiability \cite{TCC:MauRenHol04} states that if a construction $C^{\mathcal{F}}$ based on an ideal primitive $\mathcal{F}$ is indifferentiable from $\mathcal{G}$, then $C^{\mathcal{F}}$ can be used to instantiate $\mathcal{G}$ in any protocol with single-stage game. We note, however, the composition theorem does not extend to the multi-stage games, or when the adversary is resource-restricted. We refer the reader to \cite{EC:RisShaShr11} for details. 
We refer to the queries made to $C^{\mathcal{F}}/\mathcal{G}$ as construction queries and to the queries made to $\mathcal{F}/S$ as the primitive queries.

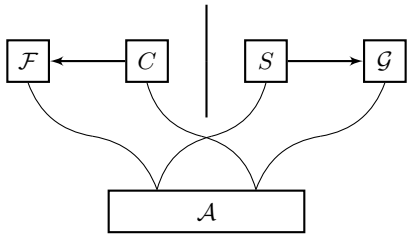
\begin{figure}[h]
\begin{center}
\input{figs/Indiff.tikz}
\end{center}
\caption{The indifferentiability notion} \vspace{-0.7cm}
\label{fig:indiff}
\end{figure}

\vspace{-5pt}
\subsubsection*{Coefficient-H Technique.}
\label{sec:h-coeff-techn}

We shall prove indifferentiability using Patarin's coefficient-H technique~\cite{SAC:Patarin08}. Fix any distinguisher $\ddv$ making $q$ queries. As the distinguisher is computationally unbounded, without loss of generality we can assume it to be deterministic \cite{INDOCRYPT:Nandi06,EC:CheSte14}. The interaction of $\ddv$ with its oracles is described by a transcript $\tau$. $\tau$ contains all the queries and the corresponding responses $\ddv$ makes during its execution. Let $\Theta$ denote the set of all possible transcripts. Let $X_{\mbox{\tt real}}$ and $X_{\mbox{\tt ideal}}$ denote the probability distribution of the transcript in the real and the ideal worlds, respectively.

\begin{lemma}{\cite{SAC:Patarin08}}
  Consider a fixed deterministic distinguisher $\ddv$. Let $\Theta$ can be partitioned into sets $\Theta_{good}$ and $\Theta_{bad}$. Suppose $\varepsilon\geq 0$ be such that for all $\tau\in\Theta_{good}$,
  \begin{align*}
    \Prob{X_{\mbox{\tt real}}=\tau}\geq (1-\varepsilon) \Prob{X_{\mbox{\tt ideal}}=\tau}
  \end{align*}
  Then
  $\Adv{\mbox {Indiff}}{(C^{\mathcal{F}},\mathcal{F}),(\mathcal{G},S^\mathcal{G})}\leq \varepsilon + \Prob{X_{\mbox{\tt ideal}}\in \Theta_{bad}}$
\end{lemma}
    
\subsubsection{Markov Inequality}
\label{sec:markov-inequality}

We recall the well known Markov inequality.
\begin{lemma}
  Let $X$ be a non-negative random variable and $a>0$ be a real number. Then it holds that
  \begin{align*}
    \Pr[X \geq a] \leq \frac{\Expect{X}}{a}
  \end{align*}
\end{lemma}

%% file: figs/Indiff.tikz
\begin{tikzpicture}[auto, node distance=1cm, >=latex']
\begin{scope}
\tikzstyle{Attacker} = [draw, rectangle, 
    minimum height=1.7em, minimum width=8em,thick]
\tikzstyle{Box} = [draw,  rectangle, 
    minimum height=1.7em, minimum width= 1.7em,thick]
\tikzstyle{to} = [->,thick]
\tikzstyle{line}= [-,thick]
\tikzstyle{dotto} = [->,dotted, thick]

\node [Box, name=F] {$\mathcal{F}$};
\node [Box, right=of F, name=C] {$C$}; 
\node [Box, right=of C, name=S] {$S$}; 
\node [Box, right=of S, name=G] {$\mathcal{G}$}; 
\draw [draw,to] (C) -- (F);
\draw [draw, to] (S) -- (G);

\coordinate (mid) at ($0.5*(C.east) + 0.5*(S.west)$);

\draw [draw, line] (mid) -- ++(0,0.75);
\draw [draw, line] (mid) -- ++(0,-0.75);

\node [Attacker, below of=mid, node distance=2cm, name=D] {$\mathcal{A}$};

\coordinate (dleftmid) at ($0.5*(D.north west) + 0.5*(D.north)$);
\coordinate (drightmid) at ($0.5*(D.north east) + 0.5*(D.north)$);

\draw (dleftmid) to[bend right, thick, dotted] ($0.5*(F.south)+0.5*(dleftmid)$);
\draw ($0.5*(F.south)+0.5*(dleftmid)$) to[->, bend left, thick, dotted] (F.south);
\draw (dleftmid) to[bend left, thick, dotted] ($0.5*(S.south)+0.5*(dleftmid)$);
\draw ($0.5*(S.south)+0.5*(dleftmid)$) to[->, bend right, thick, dotted] (S.south);

\draw (drightmid) to[bend right, thick, dotted] ($0.5*(C.south)+0.5*(drightmid)$);
\draw ($0.5*(C.south)+0.5*(drightmid)$) to[->, bend left, thick, dotted] (C.south);
\draw (drightmid) to[bend left, thick, dotted] ($0.5*(G.south)+0.5*(drightmid)$);
\draw ($0.5*(G.south)+0.5*(drightmid)$) to[->, bend right, thick, dotted] (G.south);

\end{scope}
\end{tikzpicture}

%% file: compact.tex

In Crypto 2008, Stam made the following conjecture (Conjecture 9 in ~\cite{C:Stam08}): If $C^f:\bool{m+s}\to \bool{s}$ is a compression function making $r$ calls to primitive $f:\bool{n+c}\to\bool{n}$, a collision can be found in the output of $C$ by making $q\leq 2^{\frac{nr+cr-m}{r+1}}$ queries. The conjecture was proved in two papers, the case $r=1$ was proved by Steinberger in \cite{DBLP:conf/eurocrypt/Steinberger10}, whereas the general case was proved by Steinberger, Sun and Yang in \cite{C:SteSunYan12}. The result, in our notation, is stated below.
\begin{theorem}[\cite{C:SteSunYan12}]
 Let $f_1,f_2,\ldots,f_r:\bool{n+c}\to\bool{n}$ be potentially distinct $r$ many compression functions. Let $C:\bool{m+s}\to \bool{s}$ be a domain extension algorithm making queries to $f_1,f_2,\ldots,f_r$ in the fixed order. Suppose it holds that $1\leq m\leq (n+c)r$ and $\frac{s}{2}\geq \frac{nr+cr-m}{r+1}$. There exists an adversary making at most $q=\mathcal{O}\left(r2^{\frac{nr+cr-m}{r+1}}\right)$ queries finds a collision with probability at least $\frac{1}{2}$.  
\end{theorem}


\noindent  
 In other words, if one wants to construct a hash function that achieves birthday bound collision security asymptotically, the query complexity of the attacker  must be at least  $2^{s /2}$. Then the parameters must satisfy the following equation:

\begin{align*}
  \frac{nr+cr-m}{r+1}\geq \frac{s}{2} 
\end{align*}
Next, we rearrange the equation and get
\begin{align*}
  m\leq \frac{2nr+2cr-sr-s}{2}
\end{align*}
Thus we can analyze the security-efficiency trade-off across different constructions by considering only the schemes secure (asymptotically)  up to the birthday bound and describe the efficiency by the ratio $\frac{2m}{2nr+2cr-sr-s}$. Then we argue that the optimal efficiency is reached when the parameters satisfy  
\begin{align*}
  m=\frac{2nr+2cr-sr-s}{2}
\end{align*}
Now we are ready to define compactness of hash functions based on compressing primitives.
\begin{definition}{\bf Compactness}
  \label{def:compact}
  Let $f_1,f_2,\ldots,f_r:\bool{n+c}\to\bool{n}$ be potentially distinct $r$ many compression functions. Let $C:\bool{m+s}\to \bool{s}$ be a domain extension algorithm making queries to $f_1,f_2,\ldots,f_r$ in the fixed order. We say $C$ is $\alpha$-compact if
  \begin{itemize}
  \item for all adversary $\algA$ making $q$ queries, for some constant $c_1,c_2$, it satisfies that
    \begin{align*}
      \Adv{\Coll}{C}(\algA) \leq \mathcal{O}\left( \frac{s^{c_1}r^{c_2}q^2}{2^s}\right),
    \end{align*}
  \item  \begin{align*}
  \alpha=\frac{2m}{2nr+2cr-sr-s}
  \end{align*}
  \end{itemize}
 \end{definition}

\noindent Clearly for any construction, $\alpha\leq 1$. For the rest of the paper, we consider constructions where $s=n$. Thus, we derive the value of $\alpha$ as
\begin{align*}
  \alpha=\frac{2m}{2cr+nr-n}
\end{align*}

In~\cref{sec:examples-compact}, in Examples 1 and 2  we estimate that both Merkle--Damg{\aa}rd and Merkle tree domain extenders with $2n$-to-$n$-bit compression function primitives have a compactness of $\approx 2/3$.

\subsection{Compactness of Existing Constructions}
\label{sec:examples-compact}

\begin{example}
	We consider the textbook \textbf{Merkle--Damg{\aa}rd} (MD) domain extension with length padding and fixed IV. Let the underlying function be a $2n$-to $n$-bit compression function $f$. Let the total number of calls to $f$ be $r$. At every call $n$-bits of message is processed. Assuming the length-block is of one block, the total number of message bits hashed using $r$ calls is $(r-1)c$. Hence, we get $m=(r-1)c-n$. Putting $c = n$  we compute
	
	\begin{align*}
	\alpha=\frac{2n(r-1)-2n}{2nr+nr-n}=\frac{2nr-4n}{3nr-n} < \frac{2}{3}
	\end{align*}
	
\end{example}

\begin{example} For binary \textbf{Merkle tree} with $c=n$, let the number of $f$ calls at the leaf level is $z$. Then the total number of message bit is $2nz$.  Let the total number of calls to the compression function $f$ is $r=z+z-1=2z-1$. Comparing with the number of message bits we get $m+n=(r+1)n$ which implies $m=rn$. So we calculate the compactness of Merkle tree as
	\begin{align*}
	\alpha=\frac{2rn}{3nr-n}=\frac{2r}{3r-1} < \frac{2}{3}    
	\end{align*}
\end{example}
\begin{example}
Next we consider \textbf{Shrimpton-Stam $2n$-to-$n$} compression function using three calls to $n$-to-$n$-bit function $f$. Here  $m=n$ and $c = 0$. Then $\alpha=\frac{2n}{3n-n}=1$. The \textbf{Mennink-Preneel} generalization \cite{C:MenPre12} of this construction gives $2n$-to-$n$-bit compression function making three calls to $n$-bit permutations. Thus in that case $\alpha=\frac{2n}{3n-n}=1$ as well.
\end{example}

\begin{example}\label{ex:5ary-md}
  Consider again MD domain extension with length padding and fixed IV but let the underlying function be a $5n$-to $n$-bit compression function $f$. At every (out of $r$) $f$ call $4n$-bits are processed (the rest $n$-bits are the chaining value). As
  we have one length-block, the total number of message bits hashed is $(r-1)4n$. Hence, we get $m=(r-1)4n-n$ and compute:
  \begin{align*}
    \alpha=\frac{2\times4n(r-1)-2n}{2\times 4r+nr-n} = \frac{8nr-6n}{9nr-n}\approx \frac{8}{9}
  \end{align*}
\end{example}

\begin{example}\label{ex:5ary-mt}
	The $5$-ary Merkle tree with $5z$ leaf messages  has $5nz$ bit input in total. Thus $r = \frac{3(5z-1)}{4}$ and $m = n(5z-1)$. The compactenss is given by
	\begin{align*}
	\alpha = \frac{2n(5z-1)}{2nr+nr-n} = \frac{5z-1}{3r-1} = \frac{8(5z-1)}{9(5z-1)-4} \approx \frac{8}{9}
	\end{align*}
\end{example}

%% file: collisionresistancealt.tex
 
 \section{$\cmt$ Mode with Compactness $\alpha=1$}
 \label{sec:coll-resist-gener}
 
In this section we present the  $\cmt$ domain extender. We prove its collision resistance in the random oracle model and show that it is optimally $(\alpha=1)$-compact. Our $\cmt$ mode collision-resistance-proof is valid for FIL trees. That means that our result is valid for trees of arbitrary height but once the height is fixed, all the messages queried by the adversary must correspond to a tree of that height.   We remind the reader that the majority of Merkle tree applications rely \textit{exactly} on FIL Merkle trees. \footnote{Although VIL Merkle tree exists with collision preservation proof, that is done at the cost of an extra block of Merkle-Damg{\aa}rd-type strengthening and padding schemes. As Stam's bound is derived for FIL constructions, we restrict our focus on FIL constructions only.}  
The parameter of our construction is $\ell$ which denotes the height of the tree. The construction makes $r=2^{\ell}-1$ many independent  $2n$-to-$n$-bit functions and takes input messages from the set $\bool{\mu n}$, where $\mu=2^{\ell}+2^{\ell-1}-1$. $f_{(j,b)}$ denotes the $b^{th}$ node at $j^{th}$ level. The parents of $f_{(j,b)}$ are denoted by $f_{(j-1,2b-1)}$ and $f_{(j-1,2b)}$. We use the following notations for the messages. Let $M$ be the input messages with $\mu$ many blocks of $n$-bits. The corresponding input to a leaf node $f_{(1,b)}$ is denoted by $m_{(1,2b-1)}$ and $m_{(1,2b)}$. For the internal function $f_{(j,b)}$, $m_{(j,b)}$ denotes the message that is xored with the previous chaining values to produce the input. We refer the reader to Fig.~\ref{fig:full-tree} for a pictorial view. Note, the leaves are at level $1$ and the root of the tree is at level $\ell$. The message is broken in $n$-bit blocks. $2^\ell$ many message blocks are processed at level $1$. For level $j(>1)$, $2^{\ell-j}$ many blocks are processed. The adversary $\algA$ has query access to all functions, and it makes $q$ queries in total.


\begin{figure}[htbp]
  \centering
  \begin{subfigure}[b]{0.45\textwidth}
    \centering
    \scalebox{0.70}{%
      \fbox{\procedure{$y \gets \cmt\text{ mode}(m_1, \ldots, m_{2^\ell + 2^{\ell-1}-1})$}{%
        i \gets 1, j \gets 1\\
        \pcdo  \\
        \pcind y_{1, j}  = f_{1, j}(m_i, m_{i+1})\\
        \pcind i \gets i + 2, j \gets j+1\\
        \pcwhile i < 2^\ell\\
        count \gets 2^\ell\\
        \pcfor j \pcin \{2, \ldots, \ell\}\\
        i \gets 1, s \gets count\\
        \pcind \pcdo \\
        \pcind[2] y_{j,i} = f_{j,i}(m_{s+i}\oplus y_{j-1,2i-1},\\
        \pcind[5] m_{s+i}\oplus y_{j-1,2i})\oplus y_{j-1,2i}\\
        \pcind \pcwhile i < 2^{\ell-j}\\
        count \gets count + 2^{\ell-j}\\
        \pcendfor\\
        \pcreturn y_{\ell,1}
      }}
    }
    \caption{Algorithm for computing $\cmt$ mode hash value with height $\ell$}
    \label{fig:mtree1}
  \end{subfigure}
  ~~
  \begin{subfigure}[b]{0.45\textwidth}
    \centering
    \scalebox{0.65}{\input{figs/fulltree.tikz}}
      \caption{$\cmt$ mode of height $\ell = 3$ with $2^{3}$ leaf message inputs (valid for Merkle tree), $r = 7$ compression function calls, and total of $2^{\ell}+2^{\ell-1}-1 = 11$  input blocks.}
    \label{fig:merkletree}\label{fig:full-tree}
  \end{subfigure}
  \caption{\cmt\ mode algorithm and instantiation}
\end{figure}
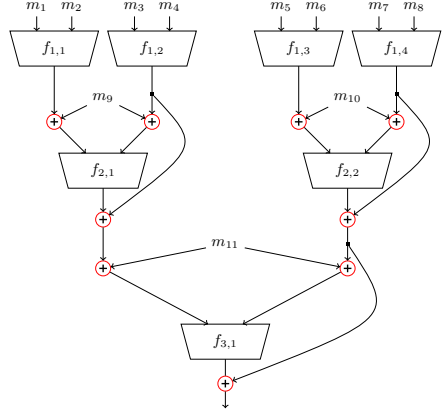

  \begin{theorem}
    \label{thm:mainbig}
    Let $\ell\geq 2$ be a natural number and $r=2^{\ell}$. Let $f:[r]\times\bool{2n}\to\booln$ be a family of functions.  Let $\algA$ be an adversary against the collision resistance of $\cmt$ mode. If the elements of $f$ are modeled as independent random oracles, then
  \begin{align*}
    \Adv{\Coll}{\cmt}(\algA^f) = \mathcal{O}\left( \frac{rn^2q^2}{2^n} \right).
  \end{align*}
  where $q$ is the number of queries $\algA$ makes to $f$ satisfying $q^2< \frac{2^n}{2e(n+1)}$. 
\end{theorem}

\subsection{Warmup: $\cmt$ mode with height $2$.}
\label{sec:fiventon}      
First, we prove the security of the case $\ell=2$. In this case $\cmt$ mode implements a $5n$-to-$n$-bit compression function with $3$ calls to $2n$-to-$n$-bit compression functions. For convenience of explanation, we refer the three functions as $f_0,f_1,f_2$ (see Fig. \ref{fig:abr-compress}). 

\begin{construction}
  Let $f_0,f_1,f_2:\bool{2n}\to\booln$ be three compression functions. We define $\cmt$ mode for $\ell=2$ as $\cmt^f:\bool{5n}\rightarrow\booln$ where
  \begin{align*}
    \cmt (m_1,m_2,m_3,m_4,m_5)= f_2\left(x_3,x_4\right)\xor f_0\left(m_5\xor f_1\left(x_1,x_2\right),m_5\xor f_2\left(x_3,x_4\right) \right) 
  \end{align*}
\end{construction}

\noindent Theorem~\ref{thm:mainbig} can be restated for this case as the following proposition.

\begin{prop}
  \label{prop:main}
  Let $f_0,f_1,f_2:\times\bool{2n}\to\booln$.  Let $\algA$ be an adversary against the collision resistance of $\cmt$. If $f_i$s  are modeled as independent random oracles, then 
  \begin{align*}
    \Adv{\Coll}{\cmt}(\algA^f) = \mathcal{O}\left(\frac{n^{2}q^2}{2^n}  \right)
  \end{align*}
  where $q$ is the maximum number of queries $\algA$ makes to the oracles $f_0,f_1,f_2$s.
\end{prop}
\textbf{Proof of Proposition~\ref{prop:main}}
The proof strategy closely follows \cite{ICALP:ShrSta08}. \\
\noindent\textsc{Moving to level-wise setting.} In general, one needs to consider the adversary making queries in some adaptive (and possibly probabilistic) manner. But for the case of $5n$-bit to $n$-bit $\cmt$, as in \cite{ICALP:ShrSta08}, we can avoid the adaptivity as $f_1$ and $f_2$ are independent random oracles. 
\begin{lemma}
  \label{lemma:adaptive-to-nonadaptive}
  For every adaptive adversary $\hat{\algA}$, there exists an adversary $\algA$ who makes level-wise queries and succeeds with same probability; $$\Adv{\Coll}{\cmt}(\hat{\algA})=\Adv{\Coll}{\cmt}(\algA).$$
\end{lemma}
\noindent\textsc{Collision Probability in the level-wise query setting}
From this point on, we assume that the adversary is provided with two lists $L_1$ and $L_2$ at the start of the game. $L_1$ and $L_2$ have $q$ uniformly sampled points and they should be considered as the responses of the queries made by the adversary to $f_1$ and $f_2$, respectively. The adversary only needs to query $f_0$.

Let $\algA$ be an adversary that can find a collision in $\cmt$. Two cases may arise. In the first case, $\algA$ can find collision in the leaf nodes ($f_1$ or $f_2$). In that case, there is a collision in either $L_1$ and $L_2$. In the other case, there is no collision among the outputs of $f_1$ or $f_2$, and the collision is generated at the final output. Let $\mathsf{Coll}_i$ denote the event that $\algA$ finds a collision in $L_i$.  Let $\mathsf{Coll}$ denote the event that $\algA$ finds a collision in $\cmt$.

\begin{align*}
  \Adv{\Coll}{\cmt}(\algA^f) \leq \Prob{\mathsf{Coll}}
  &= \Prob{\mathsf{Coll}\wedge (\mathsf{Coll}_1\vee \mathsf{Coll}_2)}+\Prob{\mathsf{Coll}\wedge \neg(\mathsf{Coll}_1\vee \mathsf{Coll}_2)}\\
  & \leq \Prob{ \mathsf{Coll}_1\vee \mathsf{Coll}_2 }+\Prob{\mathsf{Coll}\mid \neg(\mathsf{Coll}_1\vee \mathsf{Coll}_2)}\\
  &\leq \Prob{ \mathsf{Coll}_1}+ \Prob{\mathsf{Coll}_2 }+\Prob{\mathsf{Coll}\mid\neg(\mathsf{Coll}_1\vee \mathsf{Coll}_2)}.
\end{align*} 

As the functions are independent random oracles, $\Prob{\mathsf{Coll}_1}$ and $\Prob{\mathsf{Coll}_2}$ are bounded above by $\frac{q^2}{2^n}$. In the remaining, we bound the probability of the third term.

\noindent\textsc{Defining the range.} For every query $(u_i,v_i)$ made by the adversary
to $f_0$, we define the following quantity
\begin{align*}
  Y_i\eqdef\mid\{(h_1,h_2)\mid h_1\in L_1,h_2\in L_2,h_1\xor u_i=h_2\xor v_i\}\mid.
\end{align*}
where $f_0(u_i,v_i)$ is the $i^{th}$ query of the adversary.
While $Y_i$ counts the number of valid or more precisely consistent with the $\cmt$ structure pairs $(h_1, h_2)$ that were already queried to $f_1$ and $f_2$, $Y_{i}$ also denotes the number of possible $\cmt$ hash outputs produced by the adversary by making $f_0(u_i,v_i)$ query.
Notice, that $Y_i$ inputs to $f_0$ generate $Y_i$ outputs. Each of these outputs are XORed each with only one corresponding consistent $h_2$ value determined by the equation $h_1\xor u_i=h_2\xor v_i$, hence producing $Y_i$ $\cmt$ outputs on $Y_i$ consistent number inputs to $f_0$. Let $Y=\mbox{max}_i Y_i$.

\noindent\textsc{Bounding Collision by range.} Now, we show how bounding the range will help us bounding the collision probability. Let $E_i$ denotes the probability that after making the $i^{th}$ query $f_0(u_i,v_i)$ produces a collision in the output of $\cmt$. Suppose  after making $i-1$ queries, adversary is not able to produce a collision for $\cmt$. Hence, the adversary has produced $\sum_{j=1}^{i-1}Y_j$ many hash outputs. We bound the probability that $i^{th}$ query response produces a collision. 
\begin{align*}
  \prob{E_i\mid\wedge_{j=1}^{i-1}\neg E_j} \leq \frac{Y_i\sum_{j=1}^{i-1} Y_{j}}{2^n}
\end{align*}
Now we can bound the collision probability as
\begin{align*}
\Prob{\mathsf{Coll}\mid\neg(\mathsf{Coll}_1\vee \mathsf{Coll}_2)} \leq \sum_{i=1}^q  \frac{Y_i\sum_{j=1}^{i-1} Y_{j}}{2^n} \leq  \sum_{i=1}^q \sum_{j=1}^{i-1}\frac{Y^2}{2^n}\leq \frac{q^2Y^2}{2^{n+1}}
\end{align*}
We shall use the following lemma, which we prove later.
\begin{lemma}
  \label{lemma:many4xor}
  \begin{align*}
  \Prob{Y>k\mid \neg(\mathsf{Coll}_1\vee \mathsf{Coll}_2)}\leq  \frac{q^{2k}(2^n-k)!}{k!\,(2^n-1)!}  
  \end{align*}
\end{lemma}

\noindent Using Lemma~\ref{lemma:many4xor}, we get
\begin{align*}
  \Prob{\mathsf{Coll}\mid\neg(\mathsf{Coll}_1\vee \mathsf{Coll}_2)}
                                                                &\leq \Prob{\mathsf{Coll} \wedge Y\leq k\mid\neg(\mathsf{Coll}_1\vee \mathsf{Coll}_2)}\\ &\qquad +\Prob{ Y> k\mid \neg(\mathsf{Coll}_1\vee \mathsf{Coll}_2)}\\
  &\leq \frac{k^2q^2}{2^{n+1}}+\frac{q^{2k}(2^n-k)!}{k!\,(2^n-1)!} 
\end{align*}
Putting $k=n$ we get the probability as

\begin{align*}
 \Prob{\mathsf{Coll}\mid\neg(\mathsf{Coll}_1\vee \mathsf{Coll}_2)}\leq \frac{n^{2}q^2}{2^{n+1}}+ \frac{q^{2n}}{n!\, (2^n-1)\cdots (2^n-n+1)}&\approx \frac{n^{2}q^2}{2^{n+1}}+\frac{q^{2n}}{2^{n^2}}\\ &=\mathcal{O}\left(\frac{n^{2}q^2}{2^{n}}\right) 
\end{align*}
 Hence, we get the theorem. \qed

\subsubsection{Proof of Lemma~\ref{lemma:many4xor}.}
\label{sec:proof-many}
Let $(h_{i_1},h_{j_1}^\prime), (h_{i_2},h_{j_2}^\prime),\cdots,(h_{i_k},h_{j_k}^\prime)$ be the set of $k$ pairs such that each $h_{i_l}\in L_1$ and $h_{j_l}^\prime\in L_2$, and 
\begin{align*}
  h_{i_1} \xor h_{j_1}^\prime= h_{i_2}\xor h_{j_2}^\prime=\cdots=h_{i_k}\xor h_{j_k}^\prime=a~(\mbox{say})
\end{align*}

The condition $\neg(\mathsf{Coll}_1\vee \mathsf{Coll}_2)$ implies that there is no collision in $L_1$ and $L_2$. The total number of ways to choose each of $L_1$ and $L_2$ such that there is no collision is $q!{{2^n}\choose q}$. \\
Next we count the number of ways of choosing $L_1$ and $L_2$ such that the $k$ equalities get satisfied. The number of ways we can choose $i_1,i_2,\cdots,i_k$ is $q\choose k$. Fixing the order of $i_1,i_2,\cdots,i_k$, the number of ways to pair $j_1,j_2,\cdots,j_k$ is $k!{q\choose k}$. Observe that there can be $2^n$ many possible values of $a$.  Fix a value of $a$. Thus for each value of  $h_{i_l}$, there is a single value of $h_{j_l}^\prime$. Hence the total number of ways we can select $L_1,L_2$ such that the equalities get satisfied is $q!{{2^n}\choose q}\times q!{{2^n-k}\choose q}$. Hence the probability that for independently sampled $L_1$ and $L_2$,
\begin{align*}
  \Prob{Y>k\mid \neg(\mathsf{Coll}_1\vee \mathsf{Coll}_2)}=\frac{k!\left({q\choose k}\right)^22^nq!{{2^n}\choose q}\times q!{{2^n-k}\choose q}}{\left(q!{{2^n}\choose q}\right)^2}
\end{align*}
After simplification, we get the probability as
\begin{align*}
  \Prob{Y>k\mid \neg(\mathsf{Coll}_1\vee \mathsf{Coll}_2)}= \frac{(q!)^22^n(2^n-k)!}{\left((q-k)!\right)^2k!\left(2^n\right)!}\leq \frac{q^{2k}2^n(2^n-k)!}{k!\left(2^n\right)!}
\end{align*}
At the last step, we upper bound $ \frac{(q!)^2}{\left((q-k)!\right)^2}$ by $q^{2k}$. 
The lemma follows. \qed

\subsection{Proof of Theorem~\ref{thm:mainbig}}
\label{sec:proof-theorem-main}

\subsubsection{Proof Overview.}
\label{sec:proof-overview}
Now we prove the general case. We start with an overview of the proof. Unlike the case for $\ell=2$, we have to consider adaptive adversaries. Specifically, we can no longer assume that the adversary makes the queries level wise. Indeed, a query at a non-leaf level is derived from the previous chaining values (part of which is fed-forward to be xored with the output) and the messages. We can no longer ``replace'' the query without changing the chaining values. To the best of our knowledge, no proof technique achieving $2^{n/2}$ security bound asymptotically, exists in the literature for this case.

The intuition of our proof follows. Like in the previous case, our analysis focuses on the yield of a function. Informally, the yield of a query $(u,v)$ to a function $f$ is the number of chaining values created by the query. For example, consider a query $(u,v)$ made to function $f_{j,z}$, $z^{th} $ function of level $j$, and let $y$ be the output of the query. How many chaining values does this query create? A cursory inspection reveals that the number of created chaining values are the number of ``legal'' feedforward (chaining value from the previous level function $f_{j-1,2z}$) values $h$. Indeed a feedforward value $h$ can extend the chain, if there exists a chaining value $h'$ from the set of chaining values created from $f_{j-1,2z-1}$ (the other parent of $(j,z)$) such that $h'\xor u=h\xor v$. 

Naturally, if we can bound the total yield of a function (denoted as load), we can bound the probability of collision among the chaining values generated by the function. The load of a function $f_{j,z}$ gets increased in two ways. The first one is by a query made to $f_{j,z}$, as encountered in the previous section. The other one is by a query made to $f_{j',z'}$ where $j'<j$ and $(j',z')$ is in the subtree of $(j,z)$. To see why the second case holds, observe that the query to $f_{j',z'}$ increases the yield of the function, and thus creating new chaining values. Some of those newly created chaining values can be ``legal'' feedforward values for some queries already made to the next level, and thus increasing the yield of that query as well. Moreover, this in turn again creates new chaining value at the level $j'+1$. The effect continues to all the next levels and eventually affects the load of all the functions in the path to the root, including $(j,z)$.      

We bound the load of functions at each level starting from the leaves. At each level, we bound the probability of having a transcript which creates the load on a function (of that level) over a threshold amount, conditioned on the event that in none of the previous level the load exceeded the threshold.
\subsubsection{Formal Analysis.}
\label{sec:proof-theor-refthm:big}
Our formal analysis involves the transcript of the queries and the corresponding responses. Each entry of the transcript contains a query response pair, denoted by $(u,v,y)_{(j,b)}$ which indicates that $y$ is the response of the query $f_{j,b}(u,v)$. $\tau$ denotes the (partial) transcript generated after the $q$ many queries.  $Q_{(j,b)}$ denotes the set of queries made to the function $f_{(j,b)}$. $\mathcal{L}_{(j,b)}$ holds the responses.

\noindent\textsc{Yield Set} For each function $f_{(j,b)}$, we define a set $\Gamma_{(j,b)}$ holding the possible chaining values. Note, a chaining value $h \in \Gamma_{(j-1,2b)}$ can be a valid feedforward value for entry $(u,v,y)_{(j,b)}$ if there exists a matching $h'\in \Gamma_{(j-1,2b-1)}$ such that for some $m'$, it holds that $m'\xor h'=u$ and $m'\xor h=v$. Such a $m'$ can exist only if $h'\xor u= h\xor v$.
\begin{align*}
  \Gamma_{(1,b)}&\eqdef \{y \mid(u,v,y)_{(1,b)}\in \tau \}\\
  \Gamma_{(j>1,b)}&\eqdef \{y\xor h \mid (u,v,y)_{(j,b)}\in \tau, h \in \Gamma_{(j-1,2b)}, \exists h'\in \Gamma_{(j-1,2b-1)}, h'\xor u= h\xor v\}.
\end{align*}

\noindent \textsc{Feedforward set}. For each function $f_{(j,b)}$, we define a set $F_{(j,b)}$ containing the possible elements that can be used as feedforward and xored with the output of $f_{(j,b)}$ to generate valid chaining values. It is easy to verify that $F_{(j,b)}=\Gamma_{(j-1,2b)}$, where $\Gamma_{(0,b)}=\emptyset$.

Let $\mathsf{Coll}$ denotes the event that the adversary finds collision in $\cmt$ mode. Let $M=(m_{1,1},m_{1,2}\cdots,m_{1,2^{\ell}},\cdots,m_{\ell,1})$ and   $M'=(m'_{1,1},m'_{1,2}\cdots,m'_{1,2^{\ell}},\cdots,m'_{\ell,1})$ be the two distinct messages that produce the collision. We use $(u,v,y)_{(j,b)}$ and $(u',v',y')_{(j,b)}$ to be the corresponding queries made to function $f_{(j,b)}$ in the evaluation respectively.
\footnote{We assume the adversary makes all the internal queries before producing a collision. Indeed we can always add the missing queries in the transcript without significantly changing the query complexity.}

\noindent\textbf{Proper Internal Collision.}  The transcript is said to contain a \emph{proper internal collision} at $(j,b)$, if the transcript contains two distinct queries $(u,v,y)_{(j,b)}$ and $(u',v',y')_{(j,b)}$ and there exists $h,h'\in \Gamma_{(j-1,2b)}$
such that $y\xor h= y'\xor h'$. 

\begin{lemma}
  \label{lemma:crmerkle}
  Collision in tree implies a proper internal collision.
\end{lemma}

\begin{proof}
  The proof follows the Merkle tree collision resistance proof. Without loss of generality, we assume that there is no collision at the leaf. Now, consider a collision in the tree. This implies that there exist $(u,v,y)_{(\ell,1)}, (u',v',y')_{(\ell,1)}\in \tau$ and $h,h'\in \Gamma_{(\ell-1,2)}$ such that
  \begin{align*}
    y\xor h =y'\xor h'
  \end{align*}

  If $(u,v)_{(\ell,1)}\neq (u',v')_{(\ell,1)}$, then we get our proper internal collision at $(\ell,1)$, and we are done.  Otherwise $(u,v)_{(\ell,1)}=(u',v')_{(\ell,1)}$, which in turn implies $y=y'$. This implies $h=h'$. Moreover, we get $h\xor u\xor v= h'\xor u'\xor v'$
. The above two equalities give us collision in the both left and the right subtree. As $M\neq M'$, the messages differ in one of the subtrees. Repeating the above argument in the appropriate tree, we indeed find a $(j,b)$ with distinct inputs $(u,v)_{(j,b)}\neq (u',v')_{(j,b)}$. \qed
 \end{proof}

\noindent\textbf{Bounding Probabilities of a Proper Internal Collision}\\

\noindent\textsc{Yield of a query.} Consider an element $(u,v,y)_{(j,b)}\in \tau$. We define the following quantity as the yield of the query $f_{(j,b)}(u,v)$. 
\begin{align*}
  Y_{u,v,j,b}\eqdef \twopartdef{\mid \{(h_1,h_2)\mid h_1\in \Gamma_{(j-1,2b-1)},h_2\in \Gamma_{j-1,2b}, h_1\xor u=h_2\xor v\}\mid}{j>1}{1}{j=1}
\end{align*}

\noindent \textsc{Load on a function.} The load on a function $f_{(j,b)}$ is defined by the total yield of the queries made to that function.
\begin{align*}
  L_{(j,b)}\eqdef \sum_{(u_i,v_i)\in Q_{j,b}} Y_{u_i,v_i,j,b}. 
\end{align*}

Observe that if no internal collision happens at a function , the size of the yield set is the load on that function; $L_{(j,b)}= \mid\Gamma_{j,b}\mid$

For the rest of the analysis we use the variable $k$ which is equal to $(n+1)^{\frac{1}{\ell}}$. 

\noindent\textsc{Bad Events.} In this section we define the notion of bad event. We observe that with every query, the load on the functions in the tree change.  Two types of contributions to load happen with each query.
\begin{enumerate}
\item \textbf{Type I} A new $(u,v)_{(j,b)}$ query contributes to $L_{(j,b)}$. The contribution amount is $Y_{(u,v,j,b)}$.
\item \textbf{Type II} A new $(u,v)_{j',b'}$ query increases the load of $(j,b)$ where $j>j'$ and $(j',b')$ is in the sub-tree rooted at $(j,b)$.      
\end{enumerate}
$\delta^1_{(j,b)}$ and $\delta^2_{(j,b)}$ denotes the total type-I and type-II contributions to $L_{(j,b)}$ respectively. We consider the following two helping $\mathsf{Bad}$ events.
\begin{enumerate}
\item $\mathsf{Bad1}$ happens at function $(j,b)$ such that for some $(u,v,y)_{(j,b)}\in \tau$, such that $Y_{(u,v,j,b)}>k^\ell$. This event corresponds to the Type I queries. 
\item $\mathsf{Bad2}$ happens at function $(j,b)$, if $\delta^2_{(j,b)}> k^{\ell} q$.  
\end{enumerate}
$\mathsf{Bad1}_j$ and $\mathsf{Bad2}_j$ denotes the event that $\mathsf{Bad1}$ or $\mathsf{Bad2}$ respectively happens at some node at level $j$. We define $\mathsf{Bad}_j$ as $\mathsf{Bad1}_j\cup \mathsf{Bad2}_j$. Let $\mathsf{Bad}$ denote the event that for the generated transcript $\mathsf{Bad}_j$ holds for some level $j$.
\begin{align*}
  \mathsf{Bad}\eqdef \bigcup_{j} \mathsf{Bad}_j
\end{align*}
The following proposition holds from the definitions.
\begin{lemma}
  \label{lemma:badsplit}
  \begin{align*}
    \neg \mathsf{Bad}_j\implies ~\forall b\in[2^{\ell-j}]~\mbox{ it holds that } L_{(j,b)} \leq 2k^{\ell} q
  \end{align*}
  
\end{lemma}


 
\noindent\textsc{Deriving Collision Probability.} Let $\mathsf{Coll}_{j}$ denote the event of a proper internal collision at $(j,b)$ for some $b\in [2^{\ell-j}]$. 
\begin{align*}
  \Prob{\mathsf{Coll}} &\leq \Prob{\mathsf{Coll} \cup \mathsf{Bad}}\\
                       &\leq  \Prob{\mathsf{Coll}_{1}\cup \mathsf{Bad}_1}+\sum_{j>1}\Prob{(\mathsf{Coll}_{j}\cup \mathsf{Bad}_j) \cap \cap_{j'<j}\neg \mathsf{Coll}_{j'}\cap\cap_{j'< j} \neg\mathsf{Bad}_{j'}}\\
                       & \leq \Prob{\mathsf{Coll}_{1}\cup \mathsf{Bad}_1}+\sum_{j>1}\Prob{\mathsf{Bad}_{j}\cap\cap_{j'<j}\neg \mathsf{Coll}_{j'}\cap\cap_{j'<j} \neg\mathsf{Bad}_{j'}}+\\ &\qquad\qquad \sum_{j>1}\Prob{\mathsf{Coll}_{j}\cap \cap_{j'<j}\neg \mathsf{Coll}_{j'}\cap\cap_{j'\leq j} \neg\mathsf{Bad}_{j'}} \\
\end{align*}
Using the fact that $\prob{A\cap B}=\prob{A\mid B}\prob{B}\leq \prob{A\mid B}$,

\begin{align}
  \label{eq:main}
\Prob{\mathsf{Coll}} \leq& \Prob{\mathsf{Coll}_{1}\cup \mathsf{Bad}_1}+\sum_{j>1}\Prob{\mathsf{Bad}_{j}\mid\cap_{j'<j}\neg \mathsf{Coll}_{j'}\cap\cap_{j'<j} \neg\mathsf{Bad}_{j'}}+\nonumber\\ &\qquad\qquad \sum_{j>1}\Prob{\mathsf{Coll}_{j}\mid \cap_{j'<j}\neg \mathsf{Coll}_{j'}\cap\cap_{j'\leq j} \neg\mathsf{Bad}_{j'}}  
\end{align}

\noindent\textbf{Proof Sketch of Bounding $\Prob{\mathsf{Coll}_{1}\cup \mathsf{Bad}_1}$}. As all the functions are modeled as a random function, for all $b\in [2^{\ell-1}]$, we have  $\Prob{\mathsf{Coll}_{1,b}}\leq \frac{q^2}{2^n}$. Hence,
\begin{align*}
  \Prob{\mathsf{Coll}_{1}}\leq \frac{2^{\ell-1}q^2}{2^n}
\end{align*}
In order to find $\prob{\mathsf{Bad}_1}$, we recall that $F_{1,b}=\emptyset$. In other words the nothing is xored with the output of the functions at the leaf level. Hence, $Y_{(u,v,1,b)}=1$ for all $b \in [2^{\ell-1}]$ and $(u,v,y)_{1,b}\in\tau$. Hence $\Prob{\mathsf{Bad}_{1}}=0$. Hence we get,
\begin{align}
  \label{eq:leaflevel}
  \Prob{\mathsf{Coll}_{1}\cup \mathsf{Bad}_1}\leq \frac{2^{\ell-1}q^2}{2^n} 
\end{align}

\noindent \textbf{Proof Sketch of Bounding $\sum_{j>1}\Prob{\mathsf{Coll}_{j}\mid \cap_{j'<j}\neg \mathsf{Coll}_{j'}\cap\cap_{j'\leq j} \neg\mathsf{Bad}_{j'}}$}.
Fix $b\in [2^{\ell-j}]$ and thus fix a function at the $j\th$ level.  As analyzed in the previous section, given $\cap_{j'\leq j} \neg\mathsf{Bad}_{j'}$, the proper internal collision probability for $(j,b)$ is $\frac{L_{(j,b)}^2}{2^n}$. From Lemma~\ref{lemma:badsplit}, it holds that for each $b\in [2^{\ell-j}]$, $L_{(j,b)}\leq 2k^{\ell} q $. Hence for each $j>1,b\in [2^{\ell-j}]$,

\begin{align*}
  \Prob{\mathsf{Coll}_{(j,b)}\mid \cap_{j'<j}\neg \mathsf{Coll}_{j'}\cap\cap_{j'\leq j} \neg\mathsf{Bad}_{j'}}\leq \frac{4k^{2\ell}q^2}{2^n}.
\end{align*}

Taking sum over all $j>1 ,b\in [2^{\ell-j}]$,
\begin{align*}
  \sum_{j>1,b}\Prob{\mathsf{Coll}_{(j,b)}\mid \cap_{j'<j}\neg \mathsf{Coll}_{j'}\cap\cap_{j'\leq j} \neg\mathsf{Bad}_{j'}}   &\leq \sum_{j=2}^\ell \sum_{b=1}^{2^{\ell-j}} \frac{4k^{2\ell}q^2}{2^n}\\
                                                          &=\sum_{j=2}^\ell 2^{\ell-j}\times \frac{4k^{2\ell}q^2}{2^n}\\
                                                          &=\frac{2^{\ell+2}k^{2\ell}q^2}{2^n} \times \left(\sum_{j=2}^\ell  \frac{1}{2^{j}}\right)\\
%
\end{align*}
In the next step we shall use the fact that $\sum_{j=2}^\ell\frac{1}{2^{j}}< \frac{1}{2}$. Finally we get,
\begin{align}
  \label{eq:coll}
  \sum_{j>1,b}\Prob{\mathsf{Coll}_{(j,b)}\mid \cap_{j'<j}\neg \mathsf{Coll}_{j'}\cap\cap_{j'\leq j} \neg\mathsf{Bad}_{j'}} \leq \frac{2^{\ell+2}k^{2\ell}q^2}{2^{n+1}}   
\end{align}

\subsubsection{Bounding $\prob{\mathsf{Bad}}$}
\label{sec:bound-bad}



Now we bound the probabilities of the two bad events. We bound the probabilities level-wise. Let $\mathsf{Bad}1_{j,b}$ denote that $\mathsf{Bad}1$ happens at node $b$ of level $j$. Similarly, let $\mathsf{Bad}2_{j,b}$ denote that $\mathsf{Bad}2$ happens at node $b$ of level $j$. Clearly, $\mathsf{Bad}1_{j}=\cup_{b\in[2^{\ell-j}]}\mathsf{Bad}1_{j,b}$ and $\mathsf{Bad}2_{j}=\cup_{b\in[2^{\ell-j}]}\mathsf{Bad}2_{j,b}$ 

\smallskip
\noindent\textbf{Bounding $\mathbf{Bad1}_j$.}
\begin{lemma}
  \label{lemma:bad1}
  For any $(u,v,y)_{(j,b)}$ for $b\in[2^{\ell-j}]$
  \begin{align*}
     \prob{\mathbf{Bad1}_{j,b}\mid\cap_{j'<j}\neg \mathsf{Coll}_{j'}\cap\cap_{j'<j} \neg\mathsf{Bad}_{j'}}\leq 2^n\left(\frac{ek^\ell q^2}{2^n}\right)^{k^\ell}
  \end{align*}
\end{lemma}
\begin{proof}
  We bound the probability for any possible input $(u,v)_{(j,b)}$ that $ Y_{(u,v,j,b)}> k^\ell$. Fix $u\xor v=a$. Consider any entry $(u_1,v_1,y_1)_{(j-1,2b)}$ from $\tau$.  This entry contributes to $Y_{(u,v,j,b)}$ if  there exists a $h\in F_{(j-1,2b)}$  and $x\in \Gamma_{(j-1,2b-1)}$ such that $y_1\xor h\xor v=x\xor u$. Rearranging, we get that $y_1=h\xor x \xor a$. Probability of that event is $\frac{Y_{(u_1,v_1,j-1,2b)}\mid\Gamma_{(j-1,2b-1)}\mid}{2^n}$. As $\neg \textsf{Bad}_{j'}$ holds for all $j'<j$, we have $\mid\Gamma_{(j-1,2b-1)}\mid\leq k^{\ell}q$, and  $Y_{u_1,v_1,j-1,2b}\leq k^{j-1}$. Hence, the probability that $(u_1,v_1,y_1)_{(j-1,2b)}$ contributes to $Y_{u,v,j,b}$ is at most $\frac{k^{\ell+j-1}q}{2^n}$. As there are at most $q$ choices for $(u_1,v_1,y_1)_{(j-1,2b)}$ and each choice contributes one to $Y_{u,v,j,b}$,
  \begin{align*}
   \prob{Y_{u,v,j,b}>k^\ell}\leq  {q \choose k^\ell}\left(\frac{k^{\ell+j-1}q}{2^n}\right)^{k^\ell}
 \end{align*}
 Next, we use the inequality ${a \choose b} \leq \left(\frac{ea}{b}\right)^b$, where $e$ is the base of natural logarithm.
 \begin{align*}
   \prob{Y_{u,v,j,b}>k^\ell}\leq  \left(\frac{ek^{j-1}q^2}{2^n}\right)^{k^\ell}\leq \left(\frac{ek^\ell q^2}{2^n}\right)^{k^\ell}
 \end{align*}
 Now, taking union bound over all possible choice of $a$, we get that for any possible input $(u,v)$ to $f_{(j,b)}$,
 \begin{align*}
   \prob{Y_{u,v,j,b}>k^\ell} \leq  2^n\left(\frac{ek^\ell q^2}{2^n}\right)^{k^\ell}
   \tag*{\qed}
 \end{align*}
 \end{proof} 
 
\noindent \textbf{Bounding $\mathbf{Bad2}_j$.}

\begin{lemma}
  \label{lemma:bad2}
  Fix $b\in [2^{\ell-j}]$ and thus fix a function at the $j\th$ level.
  \begin{align*}
   \prob{\mathbf{Bad2}_{j,b}\mid\cap_{j'<j}\neg \mathsf{Coll}_{j'}\cap\cap_{j'<j} \neg\mathsf{Bad}_{j'}\cap \neg \mathbf{Bad1}_{j,b}}\leq \frac{2^\ell k^{\ell}q^2}{2^n} 
  \end{align*}
\end{lemma}
\begin{proof}
 Consider a query $(u,v,y)_{j',b'}$ where $(j',b')$ is in the sub-tree of $(j,b)$. As $\cap_{j'<j} \neg\mathsf{Bad}_{j'}$ holds, we argue $\neg \mathsf{Bad1}_{j'}$ holds. Thus the number of chaining value created by $(u,v,y)_{j',b'}$ query at the output of $j',b'$ is at most $k^{\ell}$, we have $Y_{u,v,j',b'}\leq k^{\ell}$.

 Next we calculate the increase in the load of the next node $f_{(j'+1,\lceil \frac{b'}{2}\rceil)}$ due to query $(u,v,y)_{j',b'}$. Consider any chaining value $h$ created due to the query $(u,v,y)_{j',b'}$.  $h$ increases the load of  $(j'+1,\lceil \frac{b'}{2}\rceil)$ if there exists  $h_1\in \Gamma_{j',b'-1}$ and $(u_1,v_1,y_1)_{j'+1,\lceil \frac{b'}{2}\rceil}\in \tau$ such that $h=h_1\xor u_1\xor v_1$. For a fixed $h_1$ and query $(u_1,v_1,y_1)_{j'+1,\lceil \frac{b'}{2}\rceil}$, probability the equation gets satisfied is $\frac{1}{2^n}$. There can be at most $\vert Q_{j'+1,\lceil \frac{b'}{2}\rceil}\vert$ many queries made to the function  $j'+1,\lceil \frac{b'}{2}\rceil$ in the transcript, implying at most $q$ many choices for candidate $(u_1,v_1,y_1)_{j'+1,\lceil \frac{b'}{2}\rceil}$.
 \begin{align*}
   \Expect{\delta^2_{(j'+1,\lceil \frac{b'}{2}\rceil)}} \leq \frac{Y_{u,v,j',b'}\left\vert \Gamma_{(j',b'-1)}\right\vert\left\vert Q_{j'+1,\lceil \frac{b'}{2}\rceil}\right\vert}{2^n}
 \end{align*}

 As  $\neg \textsf{Bad}_{j'}$ holds in the given condition, $\left\vert\Gamma_{(j',b'-1)}\right\vert=L_{(j',b'-1)} < 2k^{\ell}q$. Moreover, $Y_{u,v,j',b'}\leq k^{\ell}$; thus the expected increase in the load of $f_{(j'+1,\lceil \frac{b'}{2}\rceil)}$ is  at most $\frac{2k^{2\ell}q^2}{2^n}$.

 We extend this argument to the next levels. For a random element from $Q_{j'+1,\lceil \frac{b'}{2}\rceil}\times \Gamma_{(j',b'-1)}$ the expected number of  matched elements in $ Q_{j'+2,\lceil \frac{b'}{4}\rceil}\times \Gamma_{(j'+1,\lceil \frac{b'}{2}\rceil-1)}$ is $\frac{\left\vert\Gamma_{(j'+1,\lceil \frac{b'}{2}\rceil-1)}\right\vert\left\vert Q_{j'+2,\lceil \frac{b'}{4}\rceil}\right\vert}{\left\vert\Gamma_{(j',b'-1)}\right\vert \left\vert Q_{j'+1,\lceil \frac{b'}{2}\rceil}\right\vert}$. Using $\neg \textsf{Bad}_{j'}$ for all $j'<j$, we bound the expected increase of load for  $f_{(j'+2,\lceil \frac{b'}{4}\rceil)}$ as 
 \begin{align*}
  & \Expect{\delta^2_{(j'+2,\lceil \frac{b'}{4}\rceil)}}\\ & \leq \frac{Y_{u,v,j',b'}\mid\Gamma_{(j',b'-1)}\mid \mid Q_{j'+1,\lceil \frac{b'}{2}\rceil}\mid}{2^n}\times \frac{\mid\Gamma_{(j'+2,\lceil \frac{b'}{2}\rceil-1)}\mid \mid Q_{j'+1,\lceil \frac{b'}{4}\rceil}\mid}{\mid\Gamma_{(j',b'-1)}\mid \mid Q_{j'+1,\lceil \frac{b'}{2}\rceil}\mid}\\
                               & \leq \frac{Y_{u,v,j',b'}\mid\Gamma_{(j'+1,\lceil \frac{b'}{2}\rceil-1)}\mid \mid Q_{j'+1,\lceil \frac{b'}{4}\rceil}\mid}{2^n}\\
   & \leq \frac{2k^{2\ell}q^2}{2^n}
 \end{align*}
 Inductively extending the argument
 \begin{align*}
   \Expect{\delta^2_{(j,b)}} \leq \frac{2k^{2\ell}q^{2}}{2^{n}}.
 \end{align*}
\noindent As there $q$ many queries in the transcript, the expected total type II contribution for a function $(j,b)$ is $\frac{2^\ell k^{2\ell}q^3}{2^n}$. By using Markov inequality we get that 
\begin{align*}
  \prob{\delta^2_{(j,b)}>k^\ell q } \leq \frac{\Expect{\delta^2_{(j,b)}}}{k^\ell q} \leq \frac{2 k^{\ell}q^2}{2^n} \tag*{\qed}  
\end{align*}
\end{proof}
\subsubsection{Finishing the proof. }
\label{sec:finish}

From Lemma~\ref{lemma:badsplit}, Lemma~\ref{lemma:bad1}, and Lemma~\ref{lemma:bad2}, we bound the probability of bad as

\begin{align}
  \sum_{j>1}\Prob{\mathsf{Bad}_{j}\mid\cap_{j'<j}\neg \mathsf{Coll}_{j'}\cap\cap_{j'<j} \neg\mathsf{Bad}_{j'}}&= \sum_{j>1,b\in[2^{\ell-j}]}\left( 2^n\left(\frac{ek^\ell q^2}{2^n}\right)^{k^\ell}+\frac{2 k^\ell q^2}{2^n}\right)\\ &= \frac{2^{\ell+1}k^\ell q^2}{2^n}+  2^{\ell+n}\left(\frac{ek^\ell q^2}{2^n}\right)^{k^\ell}   \label{eq:bad}
\end{align}
From Equation~\ref{eq:main}, Equation~\ref{eq:leaflevel}, Equation~\ref{eq:coll}, and Equation~\ref{eq:bad}, we get,
\begin{align}
\Prob{\mathsf{Coll}}&\leq \frac{2^{\ell-1}q^2}{2^n}+ \frac{2^{\ell+1}k^{2\ell}q^2}{2^n}+  \frac{2^{\ell+1}k^\ell q^2}{2^n}+2^{\ell+n}\left(\frac{ek^\ell q^2}{2^n}\right)^{k^\ell}\\ & \leq \frac{2^{\ell+1}q^2(1+k^\ell+k^{2\ell})}{2^n}+2^{\ell+n}\left(\frac{ek^\ell q^2}{2^n}\right)^{k^\ell} \label{eq:final}
\end{align}
Finally, putting $k=(n+1)^{\frac{1}{\ell}}$, and assuming $q^2< \frac{2^n}{2e(n+1)}$, we get
\begin{align*}
 2^{\ell+n}\left(\frac{ek^\ell q^2}{2^n}\right)^{k^\ell}<  \frac{2^\ell e(n+1)q^2}{2^n}
\end{align*}
Putting $k^\ell=(n+1)$ in Equation \ref{eq:final},
\begin{align*}
  \Prob{\mathsf{Coll}}=\mathcal{O}\left(\frac{2^{\ell}(1+n+n^2) q^2}{2^n}\right)=\mathcal{O}\left( \frac{rn^2 q^2}{2^n}\right). 
\end{align*}
 This finishes the proof of Theorem \ref{thm:mainbig}.\qed

\begin{corollary}
The compactness of $\cmt$ is $1$. 
\end{corollary}

%% file: figs/fulltree.tikz
\begin{tikzpicture}

  \pgfmathsetmacro{\numnode}{4}
  \edef\mcount{1}
\foreach \i/\r/\j in {1/1, 3/2, 6/3, 8/4}
{
\pgfmathsetmacro{\tlev}{1}

\node[compres] (h\i) at (\i,1) {$f_{\tlev,\r}$};
\node[XOR] at (\i,-0.5) (x\i) {};
\draw[->] (h\i) -- (x\i);

\pgfmathtruncatemacro{\inp}{2*(\r-1)+1}

\node (m\inp) at ([yshift=.5cm]h\i.north west) {$m_{\inp}$};
\draw[->] (m\inp) -- (h\i.north west);
\pgfmathtruncatemacro{\nxtm}{\inp+1}

\node (m\nxtm) at ([yshift=.5cm]h\i.north east) {$m_{\nxtm}$};
\draw[->] (m\nxtm) -- (h\i.north east);
}

\foreach \j/\r/\t in {2/1/9, 7/2/10}
{
\pgfmathtruncatemacro{\tlev}{2}
\pgfmathtruncatemacro{\numsge}{8+\r}
\pgfmathtruncatemacro{\labelone}{\j-1}
\pgfmathtruncatemacro{\labeltwo}{\j+1}
\node[compres] (h\j) at (\j,-1.5) {$f_{\tlev,\r}$};
\draw[->] (x\labelone) -- (h\j.north west);
\draw[->] (x\labeltwo) -- (h\j.north east);
\node (m\numsge) at ([yshift=5mm]$ (x\labelone)!.5!(x\labeltwo) $) {$m_{\t}$};
\draw[->] (m\numsge) -- (x\labelone);
\draw[->] (m\numsge) -- (x\labeltwo);
\node[XOR] at (\j,-2.5) (x\j) {};
\node[XOR] at (\j,-3.5) (xm\j) {};
\draw[->] (x\j) -- (xm\j);
\draw[->] (h\j.south) -- (x\j);
}
\node[fill=black,inner sep=1pt] (z1) at ($ (h3.south)!.5!(x3) $) {};
\draw[->, black] (z1) .. controls (4,-1) .. (x2);
\node[fill=black,inner sep=1pt] (z2) at ($ (h8.south)!.5!(x8) $) {};
\draw[->, black] (z2) .. controls (9,-1) .. (x7);
\pgfmathtruncatemacro{\prevlabel}{7}

\node[compres] (f) at (4.5,-5) {$f_{3,1}$};
\node[XOR] at ([yshift=-.5cm]f.south) (xx) {};
\draw[->] (xm2) -- (f.north west);
\draw[->] (xm7) -- (f.north east);
\node (m11) at ([yshift=5mm]$ (xm2)!.5!(xm7) $) {$m_{11}$};
\draw[->] (m11) -- (xm2);
\draw[->] (m11) -- (xm7);

\node[fill=black,inner sep=1pt] (z3) at ($ (xm7)!.5!(x7) $) {};
\draw[->, black] (z3) .. controls (8,-4.5) .. (xx);
\draw[-] (f.south) -- (xx);
\draw[->] (xx) -- ([yshift=-1cm]f.south);

\end{tikzpicture}

%% file: indiffattack.tex
Our main result of this section is the following.
\begin{theorem}
  \label{thm:indiffattack}
  Consider the $\cmt$ mode with $\ell=2$. There exists an indifferentiability adversary $\algA$ making $\mathcal{O}(2^{\frac{n}{3}})$ many calls such that for any simulator $S$ it holds that
  \begin{align*}
    \Adv{\Indiff}{(\cmt,f),(\mathcal{G},S^\mathcal{G})}{(\algA)} \geq 1 -\epsilon 
  \end{align*}
  where $\epsilon$ is a negligible function of $n$.
\end{theorem}
Theorem~\ref{thm:indiffattack} can be extended for $\ell>2$ as well.
\subsubsection*{Principle behind the attack}
\label{sec:princ-behind-attack}
Recall the $\cmt$ with $\ell=2$ from Fig.~\ref{fig:abr-compress}. The idea is to find  collision on the input of $f_0$ for two distinct messages $m,m'$. If the adversary finds such a collision, then the output of the simulator on this input needs to be consistent with the random oracle ($\mathcal{F}$) responses on two distinct messages. That is impossible unless there is a certain relation at the output of $\mathcal{F}$, making that probability negligible.
\subsubsection{The attack}
\label{sec:attack}
The adversary $\algA$ maintains three (initially empty) query-response lists $L_0,L_1,L_2$  for the three functions $f_0,f_1,f_2$, respectively. $\algA$ chooses $2^{n/3}$ messages $(x_1^{(1)},x_2^{(1)})\in \bool{2n}$, queries to $f_1$, and adds the query-response tuple to $L_1$. Similarly, $\algA$ chooses $2^{n/3}$ messages $(x_1^{(2)},x_2^{(2)})\in \bool{2n}$, queries to $f_2$, and adds the query-response tuple to $L_2$. 
$\algA$ checks whether there exists $(x_1^{(1)},x_2^{(1)},h_1^{(1)}) \in L_1$, and  $(x_1^{(2)},x_2^{(2)},h_1^{(2)}) \in L_1$, and  $(x_3^{(1)},x_4^{(1)},h_2^{(1)}) \in L_2$, and  $(x_3^{(2)},x_4^{(2)},h_2^{(2)}) \in L_2$ such that
\begin{align}
  \label{eq:4xor}
      h_1^{(1)}\xor h_1^{(2)}\xor h_2^{(1)}\xor h_2^{(2)}=0
\end{align}
If such tuples  do not exist, $\algA$ outputs $1$ and aborts. If there is collision in the lists, $\algA$ outputs $1$ and aborts. Otherwise, it chooses a random $\hat{m}\in \booln$. The adversary sets $m=h_1^{(1)}\xor h_1^{(2)}\xor\hat{m}={h}_2^{(1)}\xor {h}_2^{(2)}\xor\hat{m}$, adversary computes $u=m\xor h_1^{(1)}= \hat{m}\xor h_1^{(2)}$ and $v=  m\xor {h}_2^{(1)}=\hat{m}\xor {h}_2^{(2)}$. Finally, adversary queries $z=f_0(u,v)$ and outputs $1$ if $z\neq \mathcal{F}(x_1^{(1)},x_2^{(1)},x_3^{(1)},x_4^{(1)},m) \xor {h}_2^{(1)}$ or $z\neq \mathcal{F}(x_1^{(2)},x_2^{(2)},x_3^{(2)},x_4^{(2)},\hat{m})\xor {h}_2^{(2)}$. Else adversary outputs $0$.\\
The full probability analysis is straightforward and skipped in this version.

%% file: indiffmt1.tex
In this section, we show that the generalized $\cmt^{+}$ mode without the additional message block at the last level is indifferentiable (up to the birthday bound) from a random oracle. For ease of explanation, we prove the result for three-level (see Fig. \ref{fig:indifftree}) balanced tree. The proof for the general case follows exactly the same idea. 
The generalized $\cmt^{+}$ mode can be viewed as the merge of two $\cmt$ mode instances, one being the left $\cmt^{+}$ branch and the other being the right branch. Both their root values are input to a final $2n$-to-$n$-bit compression function to compute the final value of the $\cmt^{+}$ tree. The  $\cmt^{+}$ tree can be either balanced or unbalanced depending on whether it uses two $\cmt$ modes of identical or distinct heights (see Fig.~\ref{fig:abrtree-indif}), respectively.


Our main result here is the following theorem. The result can be generalized to $\cmt^+$ with arbitrary height. However, the simulator description will be more detailed. For ease of explanation we consider the mode with $\ell=3$.
\begin{theorem}
  \label{thm:indiffmt1}
  Let  $f:[7]\times\bool{2n}\to\bool{n}$ be a family of random functions. Let $C^f:\bool{10n}\to\bool{n}$ be the $\cmt^+$ mode as in Fig. \ref{fig:indifftree}.  $(C^f,f)$ is $(t_S,q_S,q, \epsilon)$ indifferentiable from a random oracle $\mathcal{F}:\bool{10n}\to\booln$ where
  \begin{align*}
 \epsilon \leq \mathcal{O}\left(\frac{n^2q^2}{2^n}\right).     
  \end{align*}
where $q$ is the total number of queries made by the adversary. Moreover $t_S=\mathcal{O}(q^2)$ and $q_S=1$ 
\end{theorem}

\begin{figure}[t!]
  \centering
  \begin{subfigure}[t]{0.45\textwidth}
    \centering
    \scalebox{0.6}{\input{figs/abrindif.tikz}}
    \caption{General $\cmt^{+}$ mode}\label{fig:abrtree-indif}
    \label{figs/abrindif.tikz}
  \end{subfigure}
  ~
  \begin{subfigure}[t]{0.45\textwidth}
    \centering
    \scalebox{0.65}{\input{figs/abrplus3.tikz}}
    \caption{$\cmt^{+}$ mode with $10$ input messages}
    \label{fig:indifftree}
  \end{subfigure}
  \caption{$\cmt^{+}$ mode examples}
\end{figure}
      

\subsection{Proof of Theorem \ref{thm:indiffmt1}}
\label{sec:proof-theorem-ref}

We assume that the distinguisher $\ddv$ makes all the primitive queries corresponding the construction queries. This is without loss of generality as we can construct a distinguisher $\ddv'$ for every distinguisher $\ddv$ such that $\ddv'$ satisfies the condition. $\ddv'$ emulates $\ddv$ completely, and in particular, makes the same queries. However, at the end, for each construction queries made by $\ddv$, $\ddv'$ makes \emph{all} the (non-repeating) primitive queries required to compute the construction queries.. At the end, $\ddv'$ outputs the same decision as $\ddv$. As a result, in the transcript of $\ddv'$, all the construction query-responses, can be reconstructed from the primitive queries. Hence, it is sufficient to focus our attention on only the primitive queries and compare the distribution of outputs. If $\ddv$ makes $q_1$ many construction queries and $q_2$ many primitive queries, then $\ddv$ makes $q_1$ many construction queries and $q_2+q_1l$ many primitive queries in total where $l$ is the maximum number of primitive queries to compute $C$.

\subsubsection{The simulator}
\label{sec:simulator}
We start with the high-level overview of how the simulator $S$ works. For each $j\in [3]$, $b\in[2^{3-j}]$ the simulator maintains a list $L_{(j,b)}$. The list $L_{(j,b)}$ contains the query-response tuples for the function $f_{(j,b)}$.

\noindent\textsc{Message Reconstruction.} The main component of the simulator is the message reconstruction algorithm $\algo{FindM}$.   In the case of traditional Merkle tree, the messages are only injected in the leaf level. We have, in addition, the message injection at each (non-root) internal node. The message reconstruction in our case is slightly more involved.  

The algorithm for message reconstruction is the subroutine $\algo{FindM}$. It takes $(u_0,v_0)$, the input to $f_{(3,1)}$, as input. Let $M=m_1||m_2||\cdots||m_{10}$ be the message for which $f_{(3,1)}(u_0,v_0)$ is the hash value. Also, suppose all the intermediate queries to $f_{(j,b)} (j<3)$ has been made. In the following, we describe how the (partial) messages corresponding to chaining value $u_0$ is recovered. The other half of the message, corresponding to $v_0$, is recovered in analogous way.

Recall that there is no message injection at the final node. Hence, if all the intermediate queries related to $M$ is made by the adversary, then $m_9$ must satisfy all the following relations, $\exists (u,v,y)_{(2,1)}\in L_{(2,1)}$, such that
\begin{align*}
        y= u_0\xor v\xor m_9\qquad (m_1,m_2,u\xor m_9)\in L_{(1,1)}\qquad (m_3,m_4,v\xor m_9)\in L_{(1,2)}                    
\end{align*}

 We find a candidate $m_9$ by xoring $u_0$ with $y\xor v$ for all the (so far) recorded entries $(u,v,y)_{(2,1)}\in L_{(2,1)}$. To check the validity of the candidate, we check the other two relations. If indeed such query tuples exist, we can recover the message. 

 \noindent\textsc{Simulation of the functions.} For every non-root function $f_{(j,b)}$, $j<3$, the simulator simulates the function perfectly. Every query response is recorded in the corresponding list $L_{(j,b)}$. The simulation of $f_{(3,1)}$ is a little more involved, albeit standard in indifferentiability proof. Upon receiving a query $(u_0,v_0)$ for $f_{(3,1)}$, the simulator needs to find out whether it is the final query corresponding to the evaluation for a message $M$. Suppose, all other queries corresponding to $M$ has been made. The simulator finds $M$ using the message reconstruction algorithm. If only one candidate message $M$ is found, the simulator programs the output to be $\mathcal{F}(M)$. If the list returned by $\algo{FindM}$ is empty, then the simulator chooses a uniform random string and returns that as output. The first problem, however, arises when there are multiple candidate messages, returned by $\algo{FindM}$. This implies, there are two distinct messages $M, M'$ for both of which  $f_{(3,1)}(u_0,v_0)$ is the final query. The simulator can not program its output to both $\mathcal{F}(M)$ and $\mathcal{F}(M')$. Hence, it aborts. In that case, there is a collision at either $u$ or $v$, implying that the adversary is successful in finding a collision in $\cmt$ mode. The probability of that event can indeed be bounded by the results from the previous section. The second problem occurs in the output of non-root functions. Suppose for a $f_{(3,1)}(u_0,v_0)$ query the $\algo{FindM}$ algorithms returns an empty set. Intuitively, the simulator assumes here the adversary can not find a message $M$, for which the final query will be $f_{(3,1)}(u_0,v_0)$. Hence, the simulator does not need to maintain consistency with the Random Oracle. Now the second problem occurs, if later in the interaction, the output of some $f_{(j,b)}$ query forces a completion in the chaining value and a  message $M$ can now be recovered for which the final query will be $f_{(3,1)}(u_0,v_0)$. This will create an inconsistency of the simulator's output and the response of the Random Oracle. In the following, we bound the probability of these two events.
 
\begin{figure}[!htb]
  \centering
  \fbox{\scalebox{0.85}{
      \begin{pchstack}
        \begin{pcvstack}
          \procedure{ Procedure $S(3,1,u,v)$}{%
            \pcln \pcif (u,v,z)\in L_{(3,1)}~~\pcreturn z\\
           \pcln \mathcal{M}= \mbox{\algo{FindM}}(u,v)\\
           \pcln \pcif |\mathcal{M}| >1 \pcreturn \perp\\
            \pcln \pcif |\mathcal{M}|=0\\
          \pcln ~~z\sample\booln\\
          \pcln ~~L_{(3,1)}=L_{(3,1)}\cup (u,v,z)\\
          \pcln ~~\pcreturn z\\
           \pcln \pcendif\\
          \pcln M\leftarrow \mathcal{M} \\
          \pcln z= \mathcal{F}(M)\\
          \pcln L_{(3,1)}=L_{(3,1)}\cup (u,v,z)\\
          \pcln \pcreturn z\\
         }
        \pcvspace
       \procedure{ Procedure $S(j,b,u,v)$ where $j<3$}{%
             \pcln \pcif \exists (u,v,z)\in L_{(j,b)}\\  
          \pcln \pcind\pcreturn z\\
          \pcln\pcelse\\
          \pcln \pcind z\sample\booln\\
          \pcln \pcind \L_{(j,b)}=L_{(j,b)}\cup (u,v,z)\\
          \pcln \pcind \pcreturn z\\
          \pcln \pcendif
        }
      \end{pcvstack}
      \pchspace
      \begin{pcvstack}
         \procedure{Procedure \algo{FindM}$(u,v)$}{%
          \pccomment{Recovering message from $u$ part}\\
          \pcln \mathcal{M}_1=\emptyset\\
          \pcln \pcfor \mbox{ each } (u',v',h')\in L_{(2,1)}\\
          \pcln \pcind m_9=h'\xor u\xor v'\\
          \pcln \pcendfor\\
          \pcln \pcif \exists (m_1,m_2)\mbox{ such that } (m_1,m_2,u'\xor m_9) \in L_{(1,1)}\\
          \pcind[2] \wedge \exists (m_3,m_4)\mbox{ such that } (m_3,m_4,v'\xor m_9) \in L_{(1,2)}\\
          \pcln \pcind[2] \mathcal{M}_1 = \mathcal{M}_1\cup (m_1,m_2,m_3,m_4,m_9)\\
          \pcln\pcendif\\
           \pccomment{Recovering message from $v$ part}\\
          \pcln \mathcal{M}_2=\emptyset\\
          \pcln \pcfor \mbox{ each } (u',v',h')\in L_{2,2}\\
          \pcln \pcind m_{10}=h'\xor v \xor v'\\
          \pcln \pcendfor\\
          \pcln \pcif \exists (m_5,m_6)\mbox{ such that } (m_5,m_6,u'\xor m_{10}) \in L_{(1,3)}\\
          \pcind[2] \wedge \exists (m_7,m_8)\mbox{ such that } (m_7,m_8,v'\xor m_{10}) \in L_{(1,4)}\\
          \pcln \pcind[2] \mathcal{M}_2 = \mathcal{M}_2\cup (m_5,m_6,m_7,m_8,m_{10})\\
          \pcln\pcendif\\
           \pccomment{Combining the messages}\\
          \pcln \pcfor \mbox{ each} (m_1,m_2,m_3,m_4,m_9) \leftarrow \mathcal{M}_1\\
          \pcind[2]\wedge  \mbox{ each} (m_5,m_6,m_7,m_8,m_{10}) \leftarrow \mathcal{M}_2\\
          \pcln \mathcal{M}=\mathcal{M}\cup (m_1,m_2,\cdots,m_{10})\\
          \pcln \pcendif\\
           \pcln \pcreturn \mathcal{M}\\
        }
        %
      \end{pcvstack}
        \end{pchstack}}}      
  \caption{Description of the simulator}
  \label{fig:simulator}
\end{figure}  
\noindent The description of the simulator is given in Fig. \ref{fig:simulator}. The message reconstruction algorithm finds a candidate $m_9$ (and resp. $m_{10}$) for each entry in $L_{(2,1)}$ (and resp. $L_{(2,2)}$), and checks the validity against every entry of $L_{(1,1)}$ along with $L_{(1,2)}$ (resp. $L_{(1,3)}$ along with $L_{(1,4)}$). Thus the time complexity of message reconstruction algorithm is $\mathcal{O}(q^2)$. As the simulator invokes the message reconstruction algorithm at most once for each query, we bound $t_s=\mathcal{O}(q^2)$. Similarly, we find $q_s=1$ as the simulator has to query $\mathcal{F}$ only once per \emph{invocation}. 

\subsubsection{The bad events}
\label{sec:game-transitions}
We shall prove the theorem using the H-coefficient technique. We consider the following Bad events.

\noindent \textbf{$\mbox{\sc Bad}0$:} The set $\mathcal{M}$, returned by the message reconstruction algorithm has cardinality more that one. This implies, one can extract two message $M_1,M_2$ from the transcript such that the computation of $\cmt^+(M_1)$ and $\cmt^+(M_2)$ makes the same query to  $f_{(3,1)}$.

\noindent \textbf{$\mbox{\sc Bad}1$}: There exists an $i$, such that for the $i^{th}$ entry in the transcript $h_i=f_{(j,b)}(x_i,y_i)$ with $j<3$, there exists a message $M$ such that $C^f(M)$ can be computed from the first $i$ entries of the transcript, but can not be computed from the first $i-1$ entries. This in particular implies that there exists a $i'$ with $i'<i$, such that:
  \begin{itemize}
  \item $i'^{th}$ query is a query to $f_{(3,1)}$. $h=f_{(3,1)}(u_{i'},v_{i'})$
  \item By setting $h_i=f_{(j,b)}(x_i,y_i)$ with $\ell>0$, we create a message $M$ such that all the other chaining values of $C^f(M)$ are present in the first $i-1$ queries with $f_{(3,1)}(u_{i'},v_{i'})$ as the final query. 
  \end{itemize}


\begin{lemma}
  \label{lemma:badindiff}
  For adversary $\adv$ making $q$ many queries,  $$\Prob{\mbox{\sc Bad}} \leq \mathcal{O}\left( \frac{n^2q^2}{2^n}\right).$$    
\end{lemma}

  

  \subsubsection{Bounding $\Prob{\mbox{\sc Bad}}$}
\label{sec:bound-bad}

 We bound the probabilities of the {\sc Bad} events.
 \begin{itemize}
 \item \textbf{Case $\mbox{\sc Bad}0$:} If there is a collision in the final query of the computations for two different messages, then there is a collision in the $u$ part or $v$ part of the chain. This implies a collision in one of the $\cmt$ mode output. Hence, by Proposition \ref{prop:main}
    \begin{align*}
      \Pr[\mbox{\sc Bad}0] \leq  \mathcal{O}\left(\frac{n^2q^2}{2^n}\right)
    \end{align*}
 \item \textbf{Case $\mbox{\sc Bad}1$:} We first consider a query $f_{(j,b)}(u,v)$ with $j=2$. Let $Y_{(u,v,j,b)}$ denote the yield of this query (recall that yield of a query denotes the number of new chaining values a query creates, see page 17). As there can be at most $q$ many queries to $f_{(3,1)}$ done before this, probability that such a query raises the $\mbox{\sc Bad}1$ is bounded by $\frac{Y_{(u,v,j,b)}q}{2^n}$. Taking union bound over all the queries at $f_{(j,b)}$, the probability gets upper bounded by  $\frac{q \sum Y_{(u,v,j,b)}}{2^n}$. As we showed in the previous section this probability can be bounded by $\mathcal{O}\left(\frac{n^2q^2}{2^n}\right)$.  Finally, we consider the case of $\mbox{\sc Bad}1$ raised by some queries at the leaf level. As in the proof of collision resistance, the expected number of new chaining values created at the output by the leaf level queries is $\frac{nq^3}{2^n}$. Hence, by Markov inequality, the probability that the total number of new chaining values created is more that $q$ is at most $\frac{nq^2}{2^n}$. Finally, conditioned on the number of new chaining values be at most $q$, the probability that it matches with one of the $f_{(3,1)}$ queries is at most $\frac{q^2}{2^n}$. Hence, we get   
    \begin{equation*}
        \Pr[\mbox{\sc Bad}1] \leq \mathcal{O}\left(\frac{nq^2}{2^n}\right) 
      \end{equation*}
  \end{itemize}

\subsubsection{Good transcripts are identically distributed}
\label{sec:good-transcripts-are}

We show that the good views are identically distributed in the real and ideal worlds. Note that the simulator perfectly simulates $f$ for the internal node. The  only difference is the simulation of the final query. In case of good views, the queries to $f_0$ are of two types:
\begin{enumerate}
\item The query corresponds to the final query of a distinct message $M$, such that all the internal queries of $C^f(M)$ have occurred before. In this case, the simulator response is $\mathcal{F}(M)$. Conditioned on the rest of the transcript the output distribution remains same in  both the worlds.
 \item There is no message  $M$ in the transcript so far for which this is the final query. In this case, the response of the simulator is a uniformly chosen sample. As $\mbox{\sc Bad}1$ does not occur, the property remains true. In that case as well, the output remains same, conditioned on the rest of the transcript. 
\end{enumerate}
Hence, for all $\tau\in\Theta_{good}$
\begin{align*}
  \Prob{X_{\mbox{\tt real}}=\tau}=\Prob{X_{\mbox{\tt ideal}}=\tau}
\end{align*}
This finishes the proof of Theorem~\ref{thm:indiffmt1}.

\begin{corollary}
The compactness of $\cmt^+$ making $r$ calls to underlying $2n$-to-$n$-bit function is   $1-\frac{2}{3r-1}$. 
\end{corollary}

%% file: figs/abrindif.tikz
\begin{tikzpicture}[
level 1/.style={sibling distance=40mm}
]

\node[compres] (f) {$f$} [grow=up]
[child anchor=south]
child{node[large abr,yshift=30mm] (t) {\cmt right} edge from parent[<-]}
child{node[large abr,yshift=30mm] (s) {\cmt left} edge from parent[<-]};
\draw [decorate,decoration={brace,amplitude=10pt},xshift=-2pt,yshift=0pt]
([xshift=16mm]t.north) -- ([xshift=16mm]t.south)node [black,midway,xshift=5mm] {\footnotesize $\ell_2$};
\draw [decorate,decoration={brace,amplitude=10pt, mirror},xshift=-4pt,yshift=0pt]
([xshift=-16mm]s.north) -- ([xshift=-16mm]s.south)node [black,midway,xshift=-5mm] {\footnotesize$\ell_1$};
\draw[->] (f.south) -- ([yshift=-5mm]f.south);
\end{tikzpicture}

%% file: figs/abrplus3.tikz
\begin{tikzpicture}

\pgfmathsetmacro{\numnode}{4}
\foreach \i/\r in {1/1, 3/2, 6/3, 8/4}
{
\pgfmathsetmacro{\tlev}{1}

\node[compres] (h\i) at (\i,1) {$f_{\tlev,\r}$};
\node[XOR] at (\i,-0.5) (x\i) {};
\draw[->] (h\i) -- (x\i);

\pgfmathtruncatemacro{\inp}{2*(\r-1)+1}
\node (m\inp) at ([yshift=.5cm]h\i.north west) {$m_{\inp}$};
\draw[->] (m\inp) -- (h\i.north west);
\pgfmathtruncatemacro{\nxtm}{\inp+1}
\node (m\nxtm) at ([yshift=.5cm]h\i.north east) {$m_{\nxtm}$};
\draw[->] (m\nxtm) -- (h\i.north east);
}

\foreach \j/\r in {2/1, 7/2}
{
\pgfmathtruncatemacro{\tlev}{2}
\pgfmathtruncatemacro{\numsge}{8+\r}
\pgfmathtruncatemacro{\labelone}{\j-1}
\pgfmathtruncatemacro{\labeltwo}{\j+1}
\node[compres] (h\j) at (\j,-1.5) {$f_{\tlev,\r}$};
\draw[->] (x\labelone) -- (h\j.north west);
\draw[->] (x\labeltwo) -- (h\j.north east);
\node (m\numsge) at ($ (x\labelone)!.5!(x\labeltwo) $) {$m_{\numsge}$};
\draw[->] (m\numsge) -- (x\labelone);
\draw[->] (m\numsge) -- (x\labeltwo);
\node[XOR] at (\j,-2.5) (x\j) {};
\draw[->] (h\j.south) -- (x\j);
}
\node[fill=black,inner sep=1pt] (z1) at ($ (h3.south)!.5!(x3) $) {};
\draw[->, black] (z1) .. controls (4,-1) .. (x2);
\node[fill=black,inner sep=1pt] (z2) at ($ (h8.south)!.5!(x8) $) {};
\draw[->, black] (z2) .. controls (9,-1) .. (x7);
\pgfmathtruncatemacro{\prevlabel}{7}

\node[compres] (f) at (4.5,-4) {$f_{3,1}$};
\draw[->] (x2) -- (f.north west);
\draw[->] (x7) -- (f.north east);

\draw[->] (f.south) -- ([yshift=-1cm]f.south);

\end{tikzpicture}

%% file: applicationv1.tex
In this section, we discuss the compactness of our proposed designs, possible applications and use cases.

\subsection{Efficiency and Proof Size}
Below we discuss and compare our designs with the Merkle tree regarding efficiency of compression and authentication and proof size:  the number of openings to prove a membership of a node in a tree. 

\noindent\textsc{Efficiency of compression and authentication.}
To measure efficiency of compression we consider the amount of message (in bits) processed for a fixed tree height or a fixed number of compression function calls. As mentioned earlier, compared to a Merkle tree of height $\ell$ which  absorbs $n2^{\ell}$ message bits, the \cmt\ or $\cmt^{+}$ modes process an additional $n(2^{\ell-1}-1)$ message bits. Thus, asymptotically the number of messages inserted in our $\cmt$ (or $\cmt^{+}$) mode increases by $50\%$ compared to Merkle tree. Additionally, the cost of authentication (number of compression function calls to authenticate a node) in a Merkle tree is $\log N$ where $N = 2^{\ell}$. Here as well the \cmt\ or $\cmt^{+}$ modes compress $50\%$ more message bits compared to Merkle tree keeping the same cost of authentication as in Merkle tree as shown in~\cref{lem:proofsize}. 


\noindent \textsc{Proof Size.} We refer to the tree chaining and internal message nodes as the tree \textit{openings}. The proof size in a tree is determined by the number of openings. In a Merkle tree, the proof of membership of \textit{all} (leaf) inputs  requires $\log N$ compression function evaluations and openings each. More precisely, to prove
the \textit{membership of an arbitrary leaf input}, $\log N-1$ chaining values and one leaf input are required. 
Note that while counting the number of openings, we exclude the input for which the membership is being proved. 

\begin{lemma}
	\label{lem:proofsize}
In $\cmt$ mode, to prove the membership of any node (message block): leaf or internal, we require  $2\log N-1$ ($n$-bit) openings and $\log N$ compression function computations. 
\end{lemma}
\begin{proof}
  To prove the membership of a leaf input in the $\cmt$ mode $2(\log N -1)$ openings are required together with one leaf input. This makes a total  of $2 \log N -1$ openings. To obtain the root hash  $\log N$ computation must be computed. To prove the membership of an internal node, we need $2(\log N - 1)$ openings, excluding any openings from the level at which the  internal node resides. Additionally, one more opening is required from the level of the node. Thus, in total we need again $2\log N - 1$ openings. The number of compression calls remains $\log N$.   
\end{proof}
Compared to Merkle tree, in $\cmt^+$ the proof size increases by $\log N -1$.
Admittedly, for Merkle tree applications where the proof size is the imperative performance factor, the $\cmt^+$ modes do not provide an advantage. 

\subsection{Applications and Variants}
\noindent \textsc{\textbf{ZK-SNARKs.}} We briefly point out here the potential advantages of using the \cmt\ mode in zk-SNARKS based applications, such as Zcash. 
In a zk-SNARK \cite{EC:GGPR13} based application, increasing the number of inputs or transactions in a block means that we need to increase the size of the corresponding Merkle tree. The complexity of the proof generation process in zk-SNARK is $C\log C$ where $C$ is the circuit size of the underlying function. In $\cmt^{+}$ modes the additional messages are inserted without increasing the tree height or introducing additional compression function calls. Since the messages are only injected with xor/addition operation, this does not deteriorate the complexity of the proof generation.
Zcash uses a  Merkle tree with height $\approx 29$ and $2^{34}$ byte inputs. By using either one, \cmt\ or $\cmt^{+}$ modes, an additional of
$\approx 2^{33}$ byte inputs can be compressed {\em without} making any extra calls to the underlying compression function. Asymptotically, \cmt\ or $\cmt^{+}$
provides $50\%$ improvement in the number of maintained (in the tree structure) messages compared to a Merkle tree.\\
\noindent \textsc{Further Applications.} Our modes can be useful in applications, such as hashing on parallel processors or multicore machines: authenticating software updates, image files or videos; integrity checks of large files systems, long term archiving~\cite{archv}, content distribution, torrent systems \cite{bittorrent}, etc.


\noindent \textsc{Variants.} We continue with possible variants of utilizing the  $\cmt$ compression function in existing constructions, such as the Merkle--Damg{\aa}rd domain extender and a $5$-ary Merkle tree, and discuss their compactness and efficiency.

\noindent \textbf{Merkle--Damg{\aa}rd (MD) domain extender with $\cmt$}. When the compression function in MD is substituted by $\cmt$ ($\ell=2$) compression function, the collision resistance preservation of the original domain extender is maintained.  We obtain compactness of  $\approx 8/9$ of such an MD variant (see Section~\ref{sec:examples-compact}).

For all our modes, the high compactness allows us to absorb more messages at a fixed cost or viewed otherwise, to compress the same amount of data (e.g. as MD or Merkle tree) much cheaper. We elaborate on the latter trade-off here. To compress 1MB message with classical MD that produces a $256$-bit hash value and uses a $512$-to-$256$-bit compression function, around $31250$ calls to the underlying ($512$-to-$256$-bit) compression function are made. In contrast, $\cmt$ in MD requires just $\approx 7812$ calls to the ($512$-to-$256$) compression function, that is an impressive 4-fold cost reduction.

\noindent \textbf{ $5$-ary Merkle tree with $\cmt$}. One can naturally further construct a $5$-ary Merkle tree using $\cmt$ with compactness $< 8/9$ (see Section~\ref{sec:examples-compact}). That means to compress 1MB data with a $5$-ary $\cmt$ mode with $5n$-to-$n$-bit ($n = 256$) compression functions will require $\approx 23437$ calls to the  $512$-to-$256$-bit compression functions. Using the Merkle tree the number is $31250$ compression function calls. On the other hand, the \cmt\ and $\cmt^{+}$ modes require \textit{only} $\approx 20832$ calls.\\     


%% file: discussion.tex
The $\cmt$  mode is the first collision secure, large domain, hash function that matches Stam's bound for its parameters. The $\cmt+$ is also close to optimally efficient and  achieves the stronger indifferentiablity notion, both completed in the ideal model. 
Based on our security results we can conclude that the $\cmt^{+}$ mode is indeed the stronger proposal that achieves all the `good' function properties up to the birthday bound. Driven by practical considerations for suitable replacements of Merkle tree, the $\cmt$ mode appears to be the more natural choice. This is motivated by the fact that the majority of Merkle tree uses are indeed FIL, namely they work for messages of fixed length.

Indeed, for such FIL Merkle trees collision preservation in the standard model holds but it fails once message length variability is allowed (for that one needs to add MD strengthening and extra compression function call).
The $\cmt$ mode is proven collision secure in the  ideal model. Our result confirms the structural soundness of our domain extenders in the same fashion as the Sponge domain extender does it for the SHA-3 hash function. 

We clarify that simple modification of  $\cmt$ lead to the same security results. This variant is obtained when one uses for feed-forward the left chaining value
(instead of the right as in the $\cmt$ mode).  
The collision security proofs for this variant follows exactly the same arguments and are identical up to replacement for the mentioned values. Similarly,
an extended tree version of this constructions can be shown collision or indifferentiability secure when it is generalized in the same fashion as the $\cmt^{+}$ mode. 

An interesting practical problem is to find and benchmark concrete mode instantiations. From a theory perspective, finding compact double length constructions is an
interesting research direction.

%% file: EC_camera.bbl
\begin{thebibliography}{10}

\bibitem{bittorrent}
\url{http://bittorrent.org/beps/bep_0030.html}.

\bibitem{fullversion}
Elena Andreeva, Rishiraj Bhattacharyya, and Arnab Roy.
\newblock Compactness of hashing modes and efficiency beyond merkle tree.
\newblock {\em IACR Cryptology ePrint Archive}, 2021.

\bibitem{SCN:AndMenPre10}
Elena Andreeva, Bart Mennink, and Bart Preneel.
\newblock On the indifferentiability of the {Gr{\o}stl} hash function.
\newblock In {\em SCN 10}, volume 6280 of {\em {LNCS}}, pages 88--105.
  Springer, Heidelberg, September 2010.

\bibitem{DBLP:conf/isw/AndreevaMP10}
Elena Andreeva, Bart Mennink, and Bart Preneel.
\newblock Security reductions of the second round {SHA-3} candidates.
\newblock In Mike Burmester, Gene Tsudik, Spyros~S. Magliveras, and Ivana Ilic,
  editors, {\em Information Security - 13th International Conference, {ISC}
  2010, Boca Raton, FL, USA, October 25-28, 2010, Revised Selected Papers},
  volume 6531 of {\em Lecture Notes in Computer Science}, pages 39--53.
  Springer, 2010.

\bibitem{DBLP:journals/ijisec/AndreevaMP12}
Elena Andreeva, Bart Mennink, and Bart Preneel.
\newblock The parazoa family: generalizing the sponge hash functions.
\newblock {\em Int. J. Inf. Sec.}, 11(3):149--165, 2012.

\bibitem{zcash}
Eli Ben-Sasson, Alessandro Chiesa, Christina Garman, Matthew Green, Ian Miers,
  Eran Tromer, and Madars Virza.
\newblock Zerocash: Decentralized anonymous payments from bitcoin.
\newblock {\em IACR Cryptology ePrint Archive}, 2014:349, 2014.

\bibitem{USENIX:BCTV14}
Eli {Ben-Sasson}, Alessandro Chiesa, Eran Tromer, and Madars Virza.
\newblock Succinct non-interactive zero knowledge for a von neumann
  architecture.
\newblock In {\em USENIX Security 2014}, pages 781--796. {USENIX} Association,
  August 2014.

\bibitem{tlssign}
D~Benjamin.
\newblock Batch signing for tls.
\newblock
  \url{https://tools.ietf.org/html/draft-davidben-tls-batch-signing-02}, 2019.

\bibitem{CCS:BHKNRS19}
Daniel~J. Bernstein, Andreas H{\"u}lsing, Stefan K{\"o}lbl, Ruben Niederhagen,
  Joost Rijneveld, and Peter Schwabe.
\newblock The {SPHINCS}{$^+$} signature framework.
\newblock In {\em ACM CCS 2019}, pages 2129--2146. {ACM} Press, November 2019.

\bibitem{EPRINT:BDPA11}
Guido Bertoni, Joan Daemen, Micha{\"{e}}l Peeters, and Gilles~Van Assche.
\newblock Duplexing the sponge: single-pass authenticated encryption and other
  applications.
\newblock Cryptology ePrint Archive, Report 2011/499, 2011.
\newblock \url{http://eprint.iacr.org/2011/499}.

\bibitem{EC:BDPA13}
Guido Bertoni, Joan Daemen, Micha{\"e}l Peeters, and Gilles~Van Assche.
\newblock {Keccak}.
\newblock In {\em EUROCRYPT~2013}, volume 7881 of {\em {LNCS}}, pages 313--314.
  Springer, Heidelberg, May 2013.

\bibitem{EC:BDPV08}
Guido Bertoni, Joan Daemen, Micha{\"e}l Peeters, and Gilles {Van Assche}.
\newblock On the indifferentiability of the sponge construction.
\newblock In {\em EUROCRYPT~2008}, volume 4965 of {\em {LNCS}}, pages 181--197.
  Springer, Heidelberg, April 2008.

\bibitem{SAC:BDPV11}
Guido Bertoni, Joan Daemen, Micha{\"e}l Peeters, and Gilles {Van Assche}.
\newblock Duplexing the sponge: Single-pass authenticated encryption and other
  applications.
\newblock In {\em SAC 2011}, volume 7118 of {\em {LNCS}}, pages 320--337.
  Springer, Heidelberg, August 2012.

\bibitem{EC:BlaCocShr05}
John Black, Martin Cochran, and Thomas Shrimpton.
\newblock On the impossibility of highly-efficient blockcipher-based hash
  functions.
\newblock In {\em EUROCRYPT~2005}, volume 3494 of {\em {LNCS}}, pages 526--541.
  Springer, Heidelberg, May 2005.

\bibitem{C:BlaRogShr02}
John Black, Phillip Rogaway, and Thomas Shrimpton.
\newblock Black-box analysis of the block-cipher-based hash-function
  constructions from {PGV}.
\newblock In {\em CRYPTO~2002}, volume 2442 of {\em {LNCS}}, pages 320--335.
  Springer, Heidelberg, August 2002.

\bibitem{CHES:BKLTVV11}
Andrey Bogdanov, Miroslav Kne{\v{z}}evi{\'c}, Gregor Leander, Deniz Toz, Kerem
  Varici, and Ingrid Verbauwhede.
\newblock {Spongent}: A lightweight hash function.
\newblock In {\em CHES~2011}, volume 6917 of {\em {LNCS}}, pages 312--325.
  Springer, Heidelberg, September~/~October 2011.

\bibitem{archv}
BSI.
\newblock
  \url{https://www.bsi.bund.de/SharedDocs/Downloads/EN/BSI/Publications/TechGuidelines/TR03125/TR-03125_M3_v1_2_2.pdf}.

\bibitem{ACNS:BDKOV07}
Johannes Buchmann, Erik Dahmen, Elena Klintsevich, Katsuyuki Okeya, and Camille
  Vuillaume.
\newblock {Merkle} signatures with virtually unlimited signature capacity.
\newblock In {\em ACNS 07}, volume 4521 of {\em {LNCS}}, pages 31--45.
  Springer, Heidelberg, June 2007.

\bibitem{INDOCRYPT:BGDDK06}
Johannes Buchmann, Luis Carlos~Coronado Garc{\'\i}a, Erik Dahmen, Martin
  D{\"o}ring, and Elena Klintsevich.
\newblock {CMSS} - an improved {Merkle} signature scheme.
\newblock In {\em INDOCRYPT~2006}, volume 4329 of {\em {LNCS}}, pages 349--363.
  Springer, Heidelberg, December 2006.

\bibitem{EC:CheSte14}
Shan Chen and John~P. Steinberger.
\newblock Tight security bounds for key-alternating ciphers.
\newblock In {\em EUROCRYPT~2014}, volume 8441 of {\em {LNCS}}, pages 327--350.
  Springer, Heidelberg, May 2014.

\bibitem{C:Damgaard89b}
Ivan Damg{\aa}rd.
\newblock A design principle for hash functions.
\newblock In {\em CRYPTO'89}, volume 435 of {\em {LNCS}}, pages 416--427.
  Springer, Heidelberg, August 1990.

\bibitem{EC:GGPR13}
Rosario Gennaro, Craig Gentry, Bryan Parno, and Mariana Raykova.
\newblock Quadratic span programs and succinct {NIZKs} without {PCPs}.
\newblock In {\em EUROCRYPT~2013}, volume 7881 of {\em {LNCS}}, pages 626--645.
  Springer, Heidelberg, May 2013.

\bibitem{Haber91howto}
Stuart Haber and W.~Scott Stornetta.
\newblock How to time-stamp a digital document.
\newblock {\em Journal of Cryptology}, 3:99--111, 1991.

\bibitem{TCC:MauRenHol04}
Ueli~M. Maurer, Renato Renner, and Clemens Holenstein.
\newblock Indifferentiability, impossibility results on reductions, and
  applications to the random oracle methodology.
\newblock In {\em TCC~2004}, volume 2951 of {\em {LNCS}}, pages 21--39.
  Springer, Heidelberg, February 2004.

\bibitem{C:MenPre12}
Bart Mennink and Bart Preneel.
\newblock Hash functions based on three permutations: A generic security
  analysis.
\newblock In {\em CRYPTO~2012}, volume 7417 of {\em {LNCS}}, pages 330--347.
  Springer, Heidelberg, August 2012.

\bibitem{BartnBart}
Bart Mennink and Bart Preneel.
\newblock Efficient parallelizable hashing using small non-compressing
  primitives.
\newblock {\em Int. J. Inf. Secur.}, 15(3):285--300, 2016.

\bibitem{DBLP:conf/sp/Merkle80}
Ralph~C. Merkle.
\newblock Protocols for public key cryptosystems.
\newblock In {\em Proceedings of the 1980 {IEEE} Symposium on Security and
  Privacy, Oakland, California, USA, April 14-16, 1980}, pages 122--134. {IEEE}
  Computer Society, 1980.

\bibitem{C:Merkle89a}
Ralph~C. Merkle.
\newblock A certified digital signature.
\newblock In {\em CRYPTO'89}, volume 435 of {\em {LNCS}}, pages 218--238.
  Springer, Heidelberg, August 1990.

\bibitem{INDOCRYPT:Nandi06}
Mridul Nandi.
\newblock A simple and unified method of proving indistinguishability.
\newblock In {\em INDOCRYPT~2006}, volume 4329 of {\em {LNCS}}, pages 317--334.
  Springer, Heidelberg, December 2006.

\bibitem{SAC:Patarin08}
Jacques Patarin.
\newblock The ``coefficients {H}'' technique (invited talk).
\newblock In {\em SAC 2008}, volume 5381 of {\em {LNCS}}, pages 328--345.
  Springer, Heidelberg, August 2009.

\bibitem{EC:RisShaShr11}
Thomas Ristenpart, Hovav Shacham, and Thomas Shrimpton.
\newblock Careful with composition: Limitations of the indifferentiability
  framework.
\newblock In {\em EUROCRYPT~2011}, volume 6632 of {\em {LNCS}}, pages 487--506.
  Springer, Heidelberg, May 2011.

\bibitem{AC:RisShr07}
Thomas Ristenpart and Thomas Shrimpton.
\newblock How to build a hash function from any collision-resistant function.
\newblock In {\em ASIACRYPT~2007}, volume 4833 of {\em {LNCS}}, pages 147--163.
  Springer, Heidelberg, December 2007.

\bibitem{DBLP:journals/iacr/RivestS16}
Ronald~L. Rivest and Jacob C.~N. Schuldt.
\newblock Spritz - a spongy rc4-like stream cipher and hash function.
\newblock {\em {IACR} Cryptol. ePrint Arch.}, 2016:856, 2016.

\bibitem{C:RogSte08}
Phillip Rogaway and John~P. Steinberger.
\newblock Constructing cryptographic hash functions from fixed-key
  blockciphers.
\newblock In {\em CRYPTO~2008}, volume 5157 of {\em {LNCS}}, pages 433--450.
  Springer, Heidelberg, August 2008.

\bibitem{EC:RogSte08}
Phillip Rogaway and John~P. Steinberger.
\newblock Security/efficiency tradeoffs for permutation-based hashing.
\newblock In {\em EUROCRYPT~2008}, volume 4965 of {\em {LNCS}}, pages 220--236.
  Springer, Heidelberg, April 2008.

\bibitem{ICALP:ShrSta08}
Thomas Shrimpton and Martijn Stam.
\newblock Building a collision-resistant compression function from
  non-compressing primitives.
\newblock In {\em ICALP 2008, Part~II}, volume 5126 of {\em {LNCS}}, pages
  643--654. Springer, Heidelberg, July 2008.

\bibitem{C:Stam08}
Martijn Stam.
\newblock Beyond uniformity: Better security/efficiency tradeoffs for
  compression functions.
\newblock In {\em CRYPTO~2008}, volume 5157 of {\em {LNCS}}, pages 397--412.
  Springer, Heidelberg, August 2008.

\bibitem{DBLP:conf/eurocrypt/Steinberger10}
John~P. Steinberger.
\newblock Stam's collision resistance conjecture.
\newblock In Henri Gilbert, editor, {\em Advances in Cryptology - {EUROCRYPT}
  2010, 29th Annual International Conference on the Theory and Applications of
  Cryptographic Techniques, Monaco / French Riviera, May 30 - June 3, 2010.
  Proceedings}, volume 6110 of {\em Lecture Notes in Computer Science}, pages
  597--615. Springer, 2010.

\bibitem{C:SteSunYan12}
John~P. Steinberger, Xiaoming Sun, and Zhe Yang.
\newblock Stam's conjecture and threshold phenomena in collision resistance.
\newblock In {\em CRYPTO~2012}, volume 7417 of {\em {LNCS}}, pages 384--405.
  Springer, Heidelberg, August 2012.

\bibitem{C:Wagner02}
David Wagner.
\newblock A generalized birthday problem.
\newblock In {\em CRYPTO~2002}, volume 2442 of {\em {LNCS}}, pages 288--303.
  Springer, Heidelberg, August 2002.

\end{thebibliography}
